\documentclass[%
 reprint,
 amsmath,amssymb,
 aps,
]{revtex4-1}
\usepackage{qcircuit}
\usepackage{amsmath}
\usepackage{amssymb}
\usepackage{amsthm}
\usepackage{physics}
\usepackage{dcolumn}
\usepackage{bm}
\newtheorem{theorem}{Theorem}

\newtheorem{subtheorem}{Sub-Theorem}[theorem]

\theoremstyle{definition}
\newtheorem{defn}{Definition}

\newcommand{\qeds}{\nobreak \ifvmode \relax \else
      \ifdim\lastskip<1.5em \hskip-\lastskip
      \hskip1.5em plus0em minus0.5em \fi \nobreak
      \vrule height0.75em width0.5em depth0.25em\fi}

\usepackage{graphicx}
\usepackage{sidecap}
\usepackage[caption=false]{subfig}
\usepackage{dsfont}
\usepackage{color}

\begin{document}

\preprint{APS/123-QED}

\title{Simplifying errors by symmetry and randomisation}

\author{James Mills$^{1,2}$}
\email{J.Mills-7@sms.ed.ac.uk}
\author{Debasis Sadhukhan$^{1}$}
\author{Elham Kashefi$^{1,3}$}
 \affiliation{$^1$ School of Informatics, University of Edinburgh, 10 Crichton Street, Edinburgh EH8 9AB, Scotland}
  \affiliation{$^2$ Quandela SAS, 7 Rue Léonard de Vinci, 91300 Massy, France}
\affiliation{$^3$ Laboratoire d’Informatique de Paris 6, Centre National de la Recherche Scientifique, Sorbonne Université, 4 place Jussieu,
75252 Paris Cedex 05, France}

\date{\today}

\begin{abstract}

We present a set of methods to generate less complex error channels by quantum circuit parallelisation. The resulting errors are simplified as a consequence of their symmetrisation and randomisation. Initially, the case of a single error channel is analysed; these results are then generalised to multiple error channels. Error simplification for each method is shown to be either constant, linear, or exponential in terms of system size. Finally, example applications are provided, along with experiments run on superconducting quantum hardware and numerical simulation. These applications are: (1) reducing the sample complexity of matrix-inversion measurement error mitigation by error symmetrisation, (2) improving the effectiveness of noise-estimation circuit error mitigation by error randomisation, and (3) improving the predictability of noisy circuit performance by error randomisation.

\end{abstract}

\maketitle

\section{\label{Idea} Introduction}

Scaling up quantum hardware requires increasing system sizes while also dramatically reducing error rates. The ultimate goal of fault-tolerant quantum computing may only be reached when sufficiently large numbers of qubits possessing sufficiently small error rates are achieved. This should enable quantum error correction and the arbitrary suppression of errors \cite{gottesman_quantum_2016,aharonov_fault-tolerant_1997,gottesman_theory_1998,preskill_fault-tolerant_1998,cross_comparative_2009}. 
However, a consequence of scaling up quantum devices towards this objective is the exponential increase in complexity of errors with system size. Where error complexity is defined as the number of coefficients required to fully describe the error channel. This increasing complexity motivates the use of noise tailoring techniques, like randomised compiling \cite{wallman_noise_2016, hashim_randomized_2021}, which may be applied to transform general noise into less complex forms. Also, without quantum error correction, noise always upper bounds circuit sizes that can be successfully run on quantum devices. This can necessitate the use of smaller circuits mapped onto fewer qubits than the total number available on a given device, which motivates the use of parallelisation. Quantum circuit parallelisation, the running of circuits concurrently in parallel on a quantum device, is generally used as a means of improving quantum algorithm time efficiency \cite{cade_strategies_2020, broadbent_parallelizing_2009, bravyi_future_2022, liu_qucloud_2021,niu_how_2021, niu_enabling_2022,das_case_2019}. Parallelisation can either be performed across different subsets of qubits on one device, or else different subsets of qubits on multiple devices. One recent application of this has been in accelerating the noisy intermediate-scale (NISQ) \cite{preskill_quantum_2018} algorithm known as the variational quantum eigensolver (VQE) on a superconducting quantum device \cite{mineh_accelerating_2022}. 
The decrease in algorithm run-time is generally linear in the number of instances of the circuit run in parallel. 
While parallelisation of quantum algorithms has become a popular recourse when a large number of noisy qubits are available on current devices, this will likely only become more widespread
as systems increase in size. 
Rather than using parallelisation to achieve a speed-up in quantum algorithm run-time, we instead propose its application to reduce the complexity of errors affecting the output of a noisy circuit.

In this work we present quantum circuit parallelisation methods to reduce the complexity of errors affecting the output distribution of a given computation. The first approach involves symmetrising the errors, and the second randomising them. Either constant, linear, or exponential reductions in the complexity of errors are derived in each case, depending on the method used and the assumptions made about the errors. A schematic depicting how parallelisation can be used for error simplification is shown in Fig. \ref{figschem}. Sampling from the parallel quantum circuits and combining the computational outputs of each results in an average of the different error channels affecting the combined output distribution. 
We refer to this averaged error channel, induced by the combining of parallel circuit outputs, as the effective error channel. These parallelisation methods generate less complex effective error channels relative to the individual error channels of the parallel circuits. 
Three examples of useful applications for simplifying errors are provided, along with experiments run on superconducting quantum hardware. These applications are: reducing the sampling overhead of measurement error mitigation, increasing the effectiveness of noise-estimation circuit mitigation, and improving the predictability of noisy circuit performance.

The paper is structured as follows. Preliminary information is given in section \ref{avg}. In section \ref{sing}, we present methods by which quantum circuit parallelisation can be applied to reduce error complexity. 
The initial analysis is of single error channels. In section \ref{mult}, results are generalised to multiple channels. 
Finally, in section \ref{app} we present some applications of simplifying errors along with experiments run on hardware and numerical simulation.

\section{\label{avg} Preliminaries}

In this work, noise is assumed to be in the form of stochastic Pauli channels. This assumption is reasonable because compilation techniques like randomized compiling can be used to ensure this is the case \cite{wallman_noise_2016, hashim_randomized_2021, ville_leveraging_2021, ville_leveraging_2022, ferracin_accrediting_2019}. 
Indeed, since they are fully compatible, we envision the practical implementation of the presented methods on quantum hardware being in combination with randomized compiling. 
The randomized compiling technique involves the random insertion of Pauli gates in the input circuit, such that general error channels are transformed into stochastic Pauli noise while the overall logic of the computation remains unchanged. This reduces the Pauli-Transfer Matrix describing a general error channel acting on a set of $n$ qubits from $4^n \times 4^n$ terms to a diagonal matrix with $4^n$ terms due to the suppression of the off-diagonal terms \cite{wallman_noise_2016}. 
 An $n$-qubit stochastic Pauli channel $\mathcal{E}^P$ acting on the state $\rho$ may be written
\begin{equation}
\begin{split}
\mathcal{E}^P(\rho)&=\sum_{P\in \mathbf{P}^{\otimes n}} c_{P}P \rho P^{\dagger}.\\
\end{split}
\end{equation}
The set of Pauli coefficients represent the probability distribution over the set of $n$-qubit Pauli operators $\mathbf{P}^{\otimes n}$ for the error channel. In general, to fully characterise a general stochastic Pauli channel it is necessary to estimate $4^n$ Pauli operator coefficients, denoted $\{c_{P}\}_{P\in \mathbf{P}^{\otimes n}}$. For a stochastic Pauli channel applied immediately after the preparation of state $\rho$, the expectation value of measuring the resulting noisy state according to the operator $O$ is
\begin{equation}
\begin{split}
    O_{\rho,1}&=\text{Tr}\bigg(\sum_{P\in \mathbf{P}^{\otimes n}} c_{P}P \rho P^{\dagger}O\bigg) \\
    &= \text{Tr}(\mathcal{E}^P (\rho) O).\\
\end{split}
\end{equation}
And if this state is prepared and measured on two difference sets of qubits with different error channels $\mathcal{E}^P$ and $\mathcal{E}'^P$, then the combined expectation value is
\begin{equation}
\begin{split}
    O_{\rho,2}&=\frac{1}{2}\bigg(\text{Tr}\bigg(\sum_{P\in \mathbf{P}^{\otimes n}} c_{P}P \rho P^{\dagger}O\bigg) \\
    &\hspace{4em}+ \text{Tr}\bigg(\sum_{P '\in \mathbf{P}^{\otimes n}} c_{P'}P ' \rho P'^{\dagger}O\bigg)\bigg)\\
    &= \frac{1}{2}(\text{Tr}(\mathcal{E}^P (\rho) O) + \text{Tr}(\mathcal{E}'^P (\rho) O))\\
    &= \frac{1}{2}\text{Tr}((\mathcal{E}^P(\rho) +\mathcal{E}'^P (\rho)) O)\\
    &= \frac{1}{2}\text{Tr}((\mathcal{E}^P +\mathcal{E}'^P) (\rho) O)\\
    &= \text{Tr}(\mathcal{E}^{\text{eff}}_2 (\rho) O).\\
\end{split}
\end{equation}

Hence combining the measured outputs of the two states provides the same output as instead sampling from a single state which has been acted on by an effective error channel, $\mathcal{E}^{\text{eff}}_2$, that is the equally weighted average of the two error channels, i.e. $\mathcal{E}^{\text{eff}}_2=2^{-1}(\mathcal{E}^P +\mathcal{E}'^P)$. A probabilistic mixture of error channels of this kind is referred to as an \textit{effective error channel}. 
This can be generalised to combining the measured outputs of $N$ noisy parallel circuits, so that

\begin{equation}
\begin{split}
    O_{\rho,N}&=\frac{1}{N}\sum^N_i\text{Tr}\bigg(\sum_{P_i\in \mathbf{P}^{\otimes n}} c_{P,i}P_i \rho P_i^{\dagger}O\bigg) \\
    &= \text{Tr}\bigg(\bigg(\frac{1}{N}\sum^N_i\mathcal{E}^P_i\bigg)(\rho) O\bigg)\\
    &= \text{Tr}(\mathcal{E}^{\text{eff}}_N(\rho) O).\\
\end{split}
\end{equation}
The resulting effective error channel is the average of the $N$ error channels from the parallel circuits. In combining the outputs of the different parallel circuits to create an averaged distribution, the output state, $\rho_{\text{out}}^{\text{eff}}$, effectively being sampled is 
\begin{equation}
\begin{split}
    \rho_{\text{out}}^{\text{eff}} &= \frac{1}{N}\sum^N_i \sum_{P_i\in \mathbf{P}^{\otimes n}} c_{P,i}P_i \rho P_i^{\dagger}\\
    &=\frac{1}{N}\sum^N_i\mathcal{E}^P_i(\rho) \\
    &=\mathcal{E}^{\text{eff}}_N(\rho) .\\
\end{split}
\end{equation}
The methods we propose apply this averaging of error channels to achieve reductions, relative to the original error channels, in the total number of coefficients required to fully describe the effective error channel. This is because in the effective channel, certain subsets of Pauli operators have the same associated Pauli coefficient, and so the channel may be described using fewer coefficients.

\begin{defn}
The error complexity of an $n$ qubit stochastic Pauli channel is the cardinality of the set of Pauli operator coefficients, denoted $|\{\textbf{c}_P\}_{P\in \mathbf{P}^{\otimes n}\setminus \text{I}^{\otimes n}}|$, which is the number of distinct Pauli coefficients needed to fully describe the channel.
\end{defn}

According to this definition of error complexity, coefficients for different stochastic Pauli channel operators that are identical are counted as a single coefficient. For example, a global depolarizing channel, $D$, can be defined by a single coefficient, so the error complexity of this channel is $|\{\textbf{c}_D\}|=1$. And a general stochastic Pauli channel, $P$, has complexity $|\{\textbf{c}_P\}|=4^n-1$. 
Error complexity reduction refers to the process through which effective stochastic Pauli channels with fewer distinct coefficients are generated.
The error complexity reduction methods that follow apply parallelisation to symmetrise or randomise errors in order to create less complex effective error channels. 
For the symmetry reductions, coefficients of operators that can be mapped to each other by a given symmetry transformation become the same in the effective error channel. For the randomisation reductions, it is instead coefficients of operators which act non-trivially, that is excluding the identity operator, on the same subset of qubits that become the same in the effective channel. 

\begin{figure}
\includegraphics[width=1\linewidth]{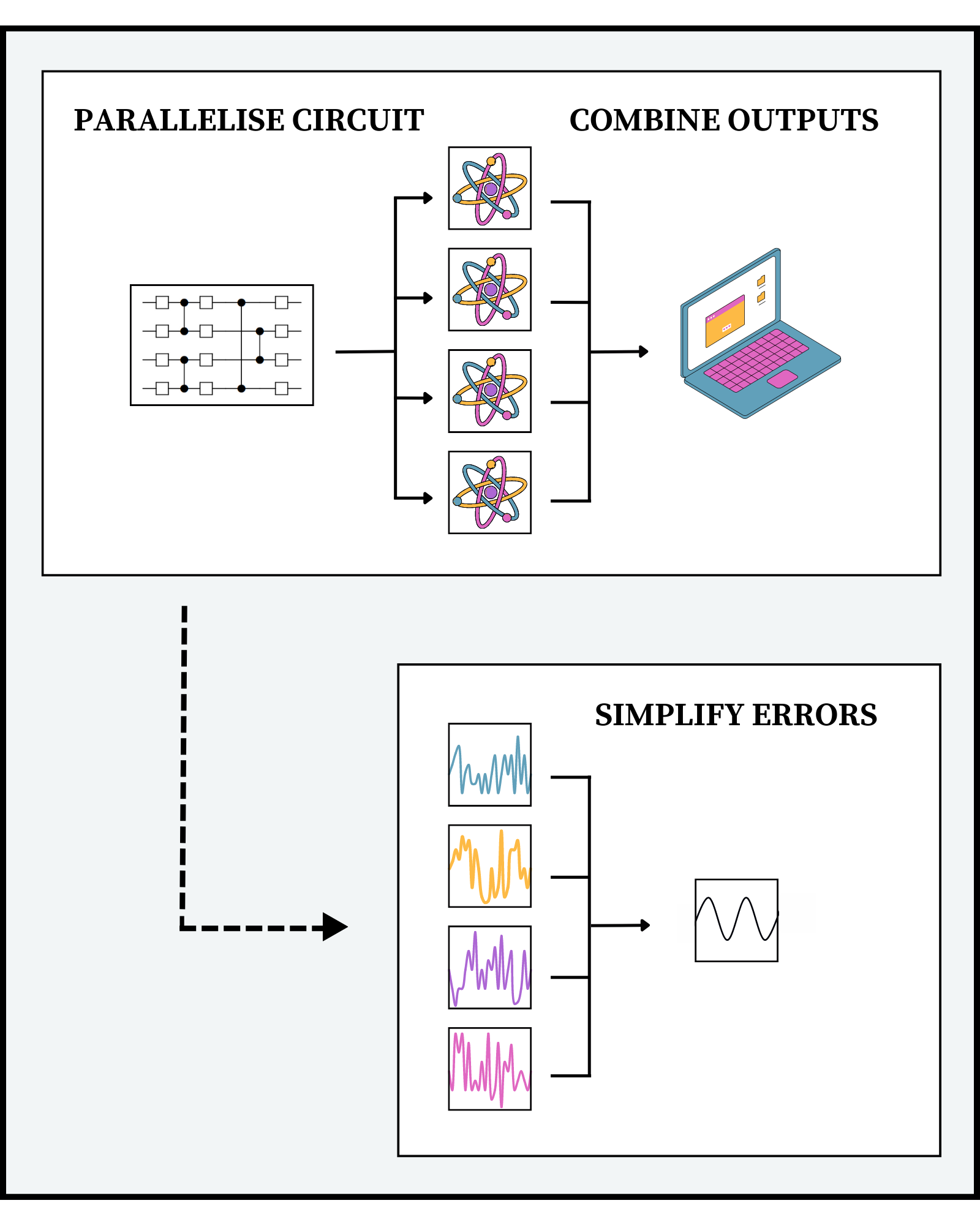}
\caption{\textit{Simplifying errors schematic.} An input quantum circuit is run in parallel on multiple sets of noisy qubits. The computational outputs of the parallel circuits are then combined such that the effects of the errors are simplified. This reduction in error complexity results from the symmetrisation and randomisation of the noisy parallel circuit errors.}
\label{figschem}
\end{figure}

\section{\label{sing} Parallelisation for a single error channel}

\begin{figure}
\scalebox{1.1}{
{\renewcommand{\arraystretch}{1.2}
\begin{tabular}{ l  c }
    \hline\hline
    Reduction type  & Effective error complexity \\ \hline
    No reduction   & $\mathcal{O}(4^n)$  \\ 
    Refl.   & $\mathcal{O}(4^n/2)$  \\ 
    Rot.      & $\mathcal{O}(4^n /n)$  \\  
    Refl. + Rot.   & $\mathcal{O}(4^n /2n)$ \\ 
    Perm.      & $\mathcal{O}(n^3/6)$ \\ 
    Rand. 1   & $\mathcal{O}(1)$ \\ 
    Rand. 2   & $\mathcal{O}(2^n)$  \\ 
    Rand. 2 + Refl. & $\mathcal{O}(2^n/2)$  \\ 
    Rand. 2 + Rot. & $\mathcal{O}(2^n/n)$ \\ 
    Rand. 2 + Refl. + Rot. & $\mathcal{O}(2^n/2n)$  \\ 
    Rand. 2 + Perm. & $\mathcal{O}(n)$  \\ \hline\hline
\end{tabular}
}
}
\begin{center}      
\caption{\textit{Effective error channel complexity summary.} The complexity of the effective error channel for each type of error complexity reduction described in Section II. For symmetry reductions, the maximum possible reductions are given without taking into account possible hardware constraints. The table entries are: no reduction, reflection symmetry, rotation symmetry, reflection and rotation symmetry, permutation symmetry, randomisation with error model (r,1), randomisation with error model (r,2),  randomisation with error model (r,2) and reflection symmetry, randomisation with error model (r,2) and rotation symmetry, randomisation with error model (r,2) with reflection and rotation symmetry, randomisation with error model (r,2) and permutation symmetry. 
}
\label{fig6}
\end{center}
\end{figure}

In this section, the circuits are assumed to have a single error channel acting at the end of the circuit, and to have a structure of alternating layers of arbitrary single qubit gates and layers of one type of symmetric multiqubit Clifford gate, like, for example, CZ gates. Circuits of this type are displayed in Fig. \ref{fig1} (a) and (b). 
Symmetry error complexity reductions are given for reflection symmetry, rotation symmetry, reflection and rotation symmetry, and permutation symmetry. Randomisation reductions are given for two different error models. Then each of the symmetry reductions is combined with randomisation for additional reductions in error complexity.

\subsection{Symmetry}

For the symmetry error complexity reductions, it is assumed that local gate errors are gate independent, and that multiqubit gate errors are gate dependent. This means that if circuit structure is preserved under different symmetry transformations applied to the qubit mappings, then the error channel remains the same. 
Here we use the term `parallelisation' to mean circuits run on the same set of qubits at different times, and these circuits are referred to as parallel circuits. The difference between parallel circuits solely being the different qubit assignments on the device topology. 
Circuit mappings are chosen in such a way that symmetries are created in the effective error channel. 
There are two types of symmetry to be considered in symmetry parallelisation error complexity reduction, \emph{device topology symmetry} and \emph{circuit symmetry}. The symmetry of the device topology, i.e. the connectivity graph for the set of qubits, determines the number of ways the nodes of the topology can be shifted around according to a particular symmetry transformation without changing the graph structure. This provides an upper bound on any possible symmetry error complexity reduction. 

Circuit symmetry relates to whether the circuit structure, the positioning of the single and multiqubit gates on the device topology disregarding gate type, is invariant under the relevant symmetry transformation. As noise is assumed to be single qubit gate independent and multiqubit gate dependent, if the positioning of single qubit gates is changed but that of multiqubit gates is the same, the error channel remains the same. 
Symmetry parallelisation with a symmetry group $S$ involves $|S|$ different versions of the error channel, which form a closed group under the given symmetry transformation. The effective channel generated by parallelising according to this symmetry acts on the state $\rho$ as 
\begin{equation}
\begin{split}
    \mathcal{E}^{\text{eff}}(\rho) &=\frac{1}{|S|}\sum_{s\in S}\mathcal{E}^s(\rho).\\
\end{split}
\end{equation}
Which is the equally weighted combination of symmetry transformed versions of the original error channel. 

Although an upper limit for the symmetry reductions is dictated by device topology, the use of these symmetries is ultimately restricted by the circuit structures it is possible to create using the native multiqubit gates available in current devices. For the noise channel to remain the same for the different circuit mappings onto the qubits, the overall circuit structure needs to be invariant under the applied symmetry transformation. For example, shown in Figure \ref{fig2} (a) is a circuit structure that is invariant under reflection symmetry. And shown in Figure \ref{fig2} (b) is a circuit invariant under rotation transformations by an even number of qubits.
For each symmetry reduction, the smallest possible multiqubit gate is used that still allows for the creation of symmetry transformation invariant circuits. For the initial symmetries only a symmetric two qubit gate acting only between pairs of qubits without overlap is needed. For permutation symmetry, a symmetric multiqubit gate acting between all qubits in the entangling layer is required. The reductions in error complexity are proportional to the number of non-trivial symmetry transformations allowed by the device topology and the circuit structure, such that the error channel remains the same.

\begin{figure}
\centering
\begin{minipage}{0.13\textwidth}
  \centering
  \includegraphics[width=0.4\linewidth]{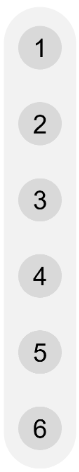}
  \begin{align*}
  \text{(a)}
  \end{align*}
  \label{fig:sub1}
\end{minipage}%
\text{\hspace{0.8cm}}
\begin{minipage}{0.1265\textwidth}
  \centering
  
  \includegraphics[width=0.4\linewidth]{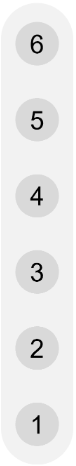}
  \begin{align*}
  \text{(b)}
  \end{align*}
  \label{fig:sub2}
\end{minipage}
\vspace{-0.5em}
\caption{\textit{Reflection symmetry.} The qubit indexing for a circuit mapping onto a set of six qubits with (a) the original qubit indexing, and (b) the reflected indexing.}
\label{fig1}
\end{figure}

\subsubsection{Reflection parallelisation}

\begin{figure}
\centering
\begin{minipage}{0.4\textwidth}
  \text{\hspace{0em}}\vspace{-1.3em}\includegraphics[width=0.9\linewidth]{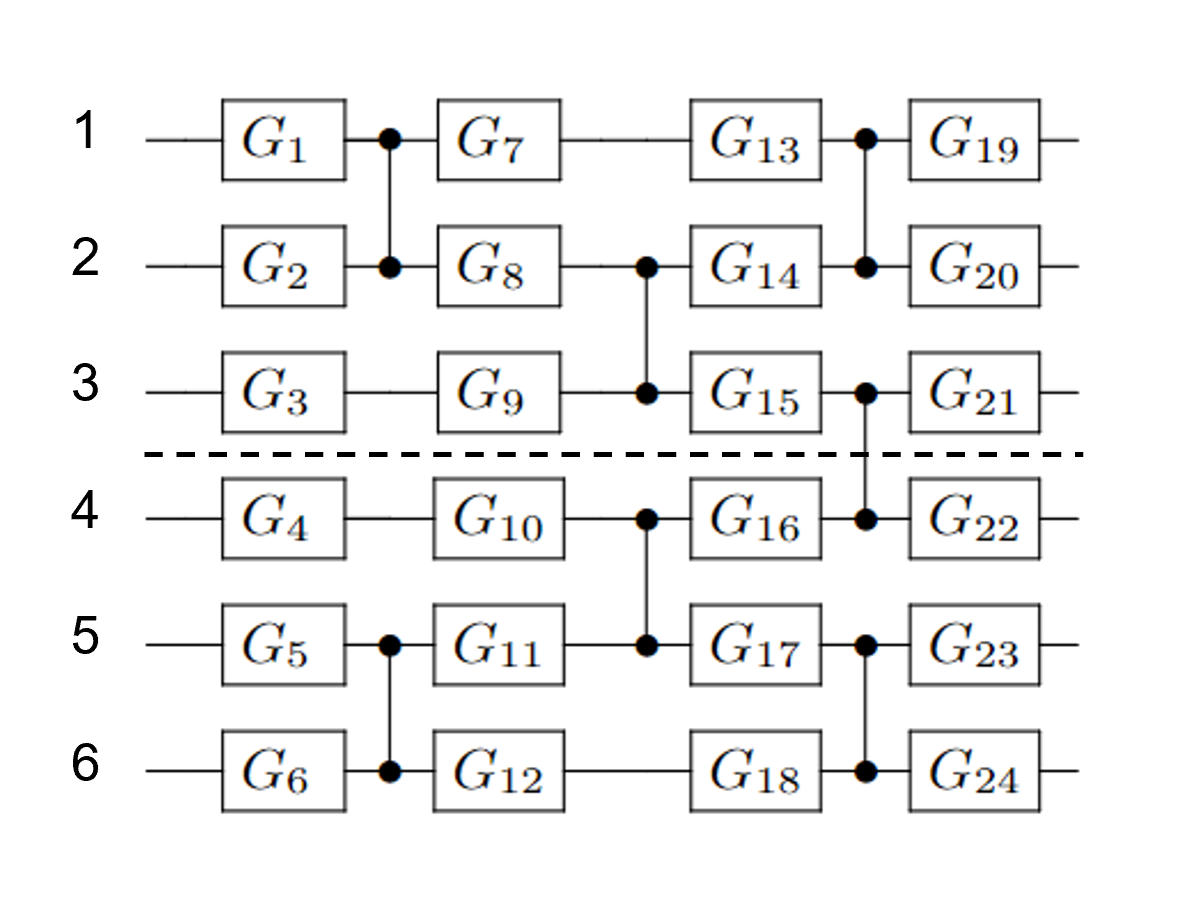}
   \begin{align*}
   \text{(a)}
   \end{align*}
  \label{fig:sub1}
\end{minipage}%
\text{\hspace{1em}}
\newline
\vspace{0.2cm}
\begin{minipage}{0.4\textwidth}
 \text{\hspace{-3.4em}} \vspace{-1.3em}\includegraphics[width=0.9\linewidth]{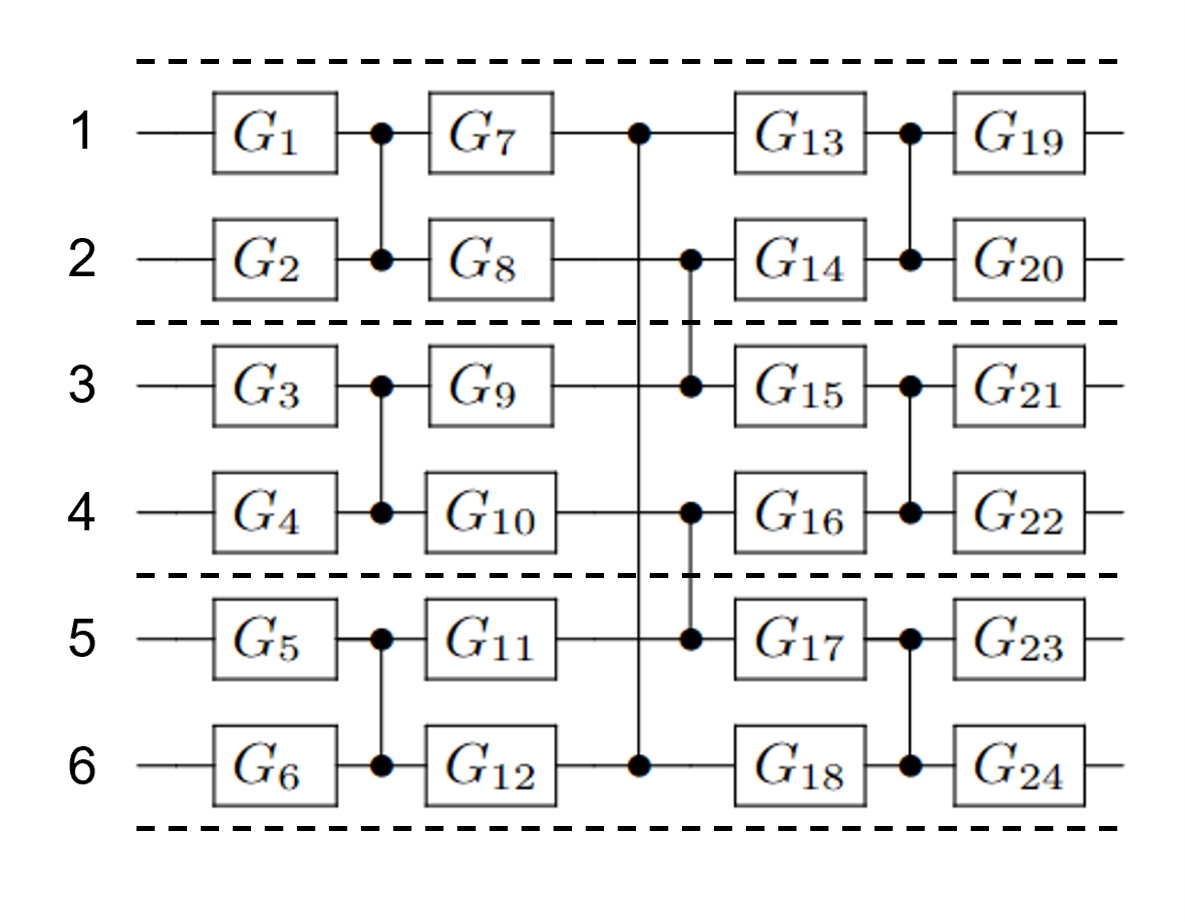}
   \begin{align*}
   \hspace{-2em}\text{(b)}
   \end{align*}
  \label{fig:sub2}
\end{minipage}
\vspace{-1em}
\caption{Example circuit diagrams with six qubit circuits with CZ gates and single qubit gates labelled $\{G_i\}_{i=1,\ldots,24}$, and with planes of circuit symmetry marked using dashed lines, for (a) reflection symmetry, and (b) rotation symmetry.}
\label{fig2}
\end{figure}

For a reflection symmetry reduction to occur, the input circuit must be reflection-symmetric. This is true if, after circuit reflection, the positions of the two qubit gates mapped onto the device topology remain unchanged, specific to the type of two qubit gate, and the positions of the single qubit gates are also unchanged, unspecific to gate type. A reflection symmetry transformation denotes the action of reversing the ordering of the qubits used for the circuit mapping. Circuits are reflection symmetry parallelised by first running the circuit and measuring the output for an initial qubit mapping, and then rerunning the circuit with the reversed mapping.
This qubit mapping is shown in Figure \ref{fig1} for a six qubit circuit. With the original qubit labelling used in Figure \ref{fig1} (a), and the reflected labelling used in Figure \ref{fig1} (b). An example of a circuit with reflection symmetry is shown in Figure \ref{fig2} (a).

If an $n$ qubit circuit is run using a linearly connected $n$ qubit subset of a quantum device, where the qubits are labelled by the $n$-tuple $(1,\text{ } 2,\text{ } \ldots,\text{ } n-1,\text{ } n)$, the reflected circuit is run using the reversed qubit ordering given by the $n$-tuple $(n,\text{ } n-1,\text{ } \ldots,\text{ } 2,\text{ } 1)$. So that for the reflection symmetry transformation, the qubit labelling is permuted according to
\begin{equation*}
\sigma=\left(\begin{array}{@{}*{20}{c@{}}}
1 & 2 & \ldots{} \text{ }&\text{ } n-1 \text{ }& \text{ }n \\
n \text{ }& \text{ }n-1 &\text{ } \ldots{} & 2 & 1   \\
\end{array}\right).
\end{equation*}

If the two instances of the circuit mapped onto the device are sampled according to an observable, and their output distributions combined with equal weight, the averaged empirical output distribution now has an associated effective error channel that is reflection symmetric. The reflection operation on the original error channel may be written in terms of its effect on the Pauli coefficients. In the original error channel $c_P = c_{P_1  P_2\ldots P_{n-1} P_n}$ denotes the coefficient of the $n$ qubit Pauli operators $P_1 \otimes P_2\otimes\ldots\otimes P_{n-1}\otimes P_n$. After the reflection operation, the coefficient for this operator in the reflected error channel becomes

\begin{equation}
\begin{split}
c_{\sigma(P)} &= c_{\sigma(P_1  P_2\ldots P_{n-1} P_n)} \\
&= c_{P_n  P_{n-1}\ldots P_{2} P_1}.\\
\end{split}
\end{equation}
Each operator coefficient in the reflected error channel is equal to the coefficient for the reflection-symmetric operator in the original channel. And so the reflected stochastic Pauli channel may be written
\begin{equation}
\begin{split}
\mathcal{E}^R(\rho)&=\sum_{P'\in \mathbf{P}^{\otimes n}} c_{\sigma(P')}P' \rho P'^{\dagger}.\\
\end{split}
\end{equation}
And the effective error channel then acts as
\begin{equation}
\begin{split}
    \mathcal{E}^{\text{eff}}(\rho)&=\frac{1}{2}(\mathcal{E} +  \mathcal{E}^{R})(\rho)\\
    &=\frac{1}{2}\bigg(\sum_{P\in \mathbf{P}^{\otimes n}} c_{P}P \rho P^{\dagger}+ \sum_{P'\in \mathbf{P}^{\otimes n}} c_{\sigma(P')}P' \rho P'^{\dagger}\bigg)\\
    &=\frac{1}{2}\sum_{P\in \mathbf{P}^{\otimes n}} (c_{P}+c_{\sigma(P)}) P \rho P^{\dagger}.\\
\end{split}
\end{equation}

Averaging the original Pauli channel with its reflection gives a constant reduction in the number of coefficients needed to describe the averaged error channel. 

\begin{theorem}

Reflection parallelisation results in a constant error complexity reduction, as the number of Pauli coefficients needed to describe the effective stochastic Pauli channel is

\begin{equation}
\begin{split}
    \normalfont|\{\textbf{c}^{ref}\}| &= 2^{2n-1} + 2^{n-1}.\\
\end{split}
\end{equation}

\end{theorem}

\begin{proof}

The coefficients of the reflection-symmetric effective error channel are composed of the average of the coefficients of Pauli operators that share reflection symmetry. And so the effective channel coefficients are all of the form: $c_{P'}=\frac{1}{2}(c_{P}+c_{\sigma(P)})$, where $P\in \mathbf{P}^{\otimes n}$ and $P'\in \{P,\sigma(P)\}$. Any Pauli operator $P$ that has reflection symmetry with itself, so-called reflection symmetry invariance where $\sigma(P)=P$, is the same in the effective error channel as it is in the original channel. There are $4^{n/2}$ reflection symmetry invariant $n$ qubit Pauli operators. 
Pauli operators that are mapped onto each other by the reflection transformation are reflection-symmetric operators. These share a single Pauli coefficient in the effective error channel, which is the average of the two coefficients. The number of operators that share reflection symmetry with a different operator is $4^{n}-4^{n/2}$. Therefore the total number of distinct Pauli coefficients in the effective error channel is

\begin{equation}
\begin{split}
\frac{1}{2}(4^{n}-4^{n/2}) + 4^{n/2} &= 2^{2n-1} + 2^{n-1}.\\
\end{split}
\end{equation}

\end{proof} 

The previous expression for the total number of distinct coefficients in the effective error channel may also be expressed in terms of the contributions from different Pauli operators for each of the different subsets of qubits. In this alternative form the expression becomes

\begin{equation}
\begin{split}
|\{\textbf{c}^{ref}\}| &= \frac{1}{2}\bigg(\sum_{i=0}^n {n \choose i}3^i - \sum_{j=0}^{n/2} {n/2 \choose j}3^j\bigg) + \sum_{k=0}^{n/2} {n/2 \choose k}3^k\\
&=\frac{1}{2}\bigg(\sum_{i=0}^n {n \choose i}3^i + \sum_{j=0}^{n/2} {n/2 \choose j}3^j\bigg). \\
\end{split}
\end{equation}

Where the combination terms denote different patterns of qubits from the subset, and the power of 3 terms the number of different Pauli operators that act non-trivially on that subset of qubits. 

\subsubsection{Rotation parallelisation}

\begin{figure*}
    \centering
    \begin{minipage}{0.25\textwidth}
        \centering
        \includegraphics[width=0.9\textwidth]{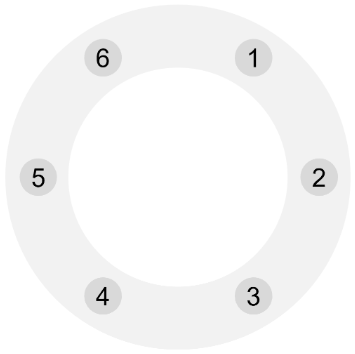} 
        \begin{align*}
        \text{(a)}    
        \end{align*}
    \end{minipage}
    \text{\hspace{0.5cm}}
    \begin{minipage}{0.25\textwidth}
        \centering
        \includegraphics[width=0.9\textwidth]{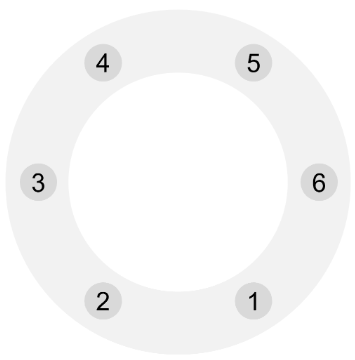} 
        \begin{align*}
        \text{(b)}    
        \end{align*}
    \end{minipage}
    \text{\hspace{0.5cm}}
    \begin{minipage}{0.25\textwidth}
        \centering
        \includegraphics[width=0.9\textwidth]{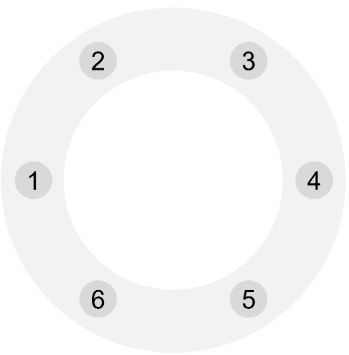} 
        \begin{align*}
        \text{(c)}    
        \end{align*}
    \end{minipage}
\caption{\textit{Rotation symmetry.} The qubit indexing for a circuit mapping onto a closed loop subset of six qubits, with (a) the original indexing, (b) the first rotation symmetry indexing, and (c) the second rotation symmetry indexing.}
\label{fig3}
\end{figure*}

A similar reduction to that found using reflection symmetry can be realised instead by applying rotation symmetry for this purpose. A locally connected loop of qubits forms a regular $n$-sided polygon, or $n$-gon, on the device topology, with vertices being qubit locations and edges the connections between the qubits. If such a shape has $n$ edges then it has order $n$ rotational symmetry. The input circuit is parallelised using mappings onto such a closed loop of qubits, this creates rotation symmetries in the resulting effective error channel. Examples of rotation parallelisation mapping for a six qubit circuit are shown in Figure \ref{fig3}. We assume simultaneous application of two qubit gates on overlapping pairs of qubits is not allowed. This limitation still allows rotation symmetry reductions but the reduction is $n/2$ rather than $n$. It is required that adjacent pairs of qubits are connected by the same symmetric two qubit gate, such that shifting the circuit around the loop by even numbers of qubits does not change the pairs of qubits between which two qubits gates are applied on the topology. Given that the positioning of the two qubit gates is the same, mapping the circuit onto the topology according to these rotation symmetry transformations does not change the error channel. 

For the different rotation symmetry mappings onto the device topology, an $n$ qubit circuit is run on a closed loop of $n$ qubits with the circuit configurations denoted by the set of $n$-tuples $(2k + 1 \text{ mod } n+1,\text{ } 2k + 2 \text{ mod } n+1,\text{ } \ldots, \text{ }2k+n-1 \text{ mod } n+1,\text{ } 2k+n \text{ mod } n+1)$ for $k\in\{0,\text{ }1,\text{ }\ldots,\text{ }\frac{n}{2}-1\}$. There are $\frac{n}{2}$ different rotation symmetry mappings of the circuit onto the loop of qubits, with the circuit configurations for each of these given by the different $n$-tuples. So that the qubit labelling is permuted according to
\begin{equation*}
\pi=\left(\begin{array}{@{}*{20}{c@{}}}
1 & 2 & \ldots{} \text{ }&\text{ } n-1 \text{ }& \text{ }n \\
n-1 \text{ }& \text{ }n &\text{ } \ldots{} & n-3 & n-2   \\
\end{array}\right).
\end{equation*}
And the same stochastic Pauli channel acts on the loop of qubits regardless of how the circuits are configured due to the two qubit gate positions remaining unchanged for all mappings. The effective error channel then acts on $\rho$ as
\begin{equation}
\begin{split}
    \mathcal{E}^{\text{eff}}(\rho) &=\frac{2}{n}\sum_{P\in \mathbf{P}^{\otimes n}}\sum_{i=0}^{\frac{n}{2}-1} c_{\pi^i(P)} P \rho P^{\dagger},\\
\end{split}
\end{equation}
and there is a linear reduction in the error complexity relative to the original error channels.

\begin{theorem}
Rotation parallelisation results in a linear error complexity reduction, as the number of Pauli coefficients needed to describe the effective stochastic Pauli channel is

\begin{equation}
\begin{split}
\normalfont|\{\textbf{c}^{rot}\}| &= \frac{2^{2n+1}-32}{n}+16\\
\end{split}
\end{equation}

\end{theorem}

\begin{proof} 

A rotation symmetric $n$ qubit circuit with non-overlapping symmetric two qubit gates in its entangling gate layers has $\frac{n}{2}$ planes of symmetry. The circuit is mapped onto the qubits according to each of these symmetries and the outputs combined. The resulting effective error channel has coefficients of the form: $c_{P'}=\frac{2}{n}\sum_{i=0}^{\frac{n}{2}-1} c_{\pi^i(P)}$, where $P\in \mathbf{P}^{\otimes n}$ and $P'\in \{\pi^i(P)\}_i$. Pauli operators that are rotation symmetry transformation invariant, and are left unchanged by rotating through an even number of qubits around the loop so that $\pi^i(P)=P$ for all $i\in\{0,1,\ldots,\frac{2}{n}-1\}$, can only be mapped to themselves and so do not contribute to the reduction. Operators that change under rotation can be mapped to $\frac{n}{2}$ other operators. There are three subsets of qubits on the loop that are rotation symmetry transformation invariant, the subset of all $n$ qubits and two subsets of $\frac{n}{2}$ qubits. The two subsets of $\frac{n}{2}$ qubits are the subset of even number labelled qubits on the device topology and the subset of odd number labelled qubits. There are $3^2$ rotationally invariant Pauli operators which act non-trivially, that is excluding the identity operator, on the $n$ qubit subset. This is because the set of Pauli operators that act non-trivially on two qubits is $\mathbf{G}= \{X,Y,Z\}^{\otimes 2}$, and the set of rotationally invariant operators is then $G^{\otimes \frac{n}{2}}$ for $G \in \mathbf{G}$. The cardinality of the set $\mathbf{G}$ is $3^2$.
The $n$ qubit identity operator is also invariant under rotation. And for each of the two $n/2$ qubit subsets consisting of alternating qubits on the loop there are three rotation invariant Pauli operators - one for each type of local Pauli operator. These are given by the two sets of operators $\{(P\otimes X)^{\otimes \frac{n}{2}}\}_{P \in \{X,Y,Z\}}$ and $\{(I\otimes P)^{\otimes \frac{n}{2}}\}_{P \in \{X,Y,Z\}}$. So the number of Pauli operators which do not remain invariant under the symmetry transformation is then $4^n - 3^2 - 1 - 2\cdot3$. This means that the number of coefficients is for the rotationally symmetric effective error channel is

\begin{equation}
\begin{split}
|\{\textbf{c}^{rot}\}| &= \frac{2}{n}(\sum_{i=0}^n {n \choose i}3^i - 3^2 - 2\cdot3^{1} - 1) \\
&+ 3^2 + 2\cdot3^{1} + 1\\
&= \frac{2^{2n+1}-32}{n}+16.\\
\end{split}
\end{equation}

\end{proof}

This result assumes that the input circuit has $n/2$ rotation symmetries, the most possible when using non-overlapping symmetric two qubit gates in the gate entangling layers. 

\subsubsection{Reflection and rotation parallelisation}

\begin{figure*}
    \centering
    \begin{minipage}{0.25\textwidth}
        \centering
        \includegraphics[width=0.9\textwidth]{loop1.png} 
        \begin{align*}
        \text{(a)}    
        \end{align*}
        \vspace{1em}
    \end{minipage}
    \text{\hspace{0.5cm}}
    \begin{minipage}{0.25\textwidth}
        \centering
        \includegraphics[width=0.9\textwidth]{loop2.png} 
        \begin{align*}
        \text{(b)}    
        \end{align*}
        \vspace{1em}
    \end{minipage}
    \text{\hspace{0.5cm}}
    \begin{minipage}{0.25\textwidth}
        \centering
        \includegraphics[width=0.9\textwidth]{loop3.png} 
        \begin{align*}
        \text{(c)}    
        \end{align*}
        \vspace{1em}
    \end{minipage}
    \begin{minipage}{0.25\textwidth}
        \centering
        \includegraphics[width=0.9\textwidth]{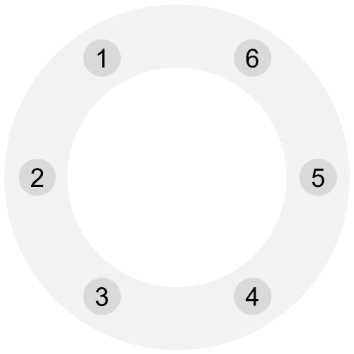} 
        \begin{align*}
        \text{(d)}    
        \end{align*}
    \end{minipage}
    \text{\hspace{0.5cm}}
    \begin{minipage}{0.25\textwidth}
        \centering
        \includegraphics[width=0.9\textwidth]{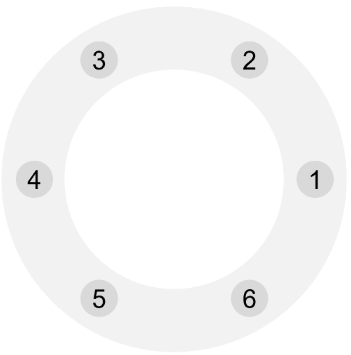} 
        \begin{align*}
        \text{(e)}    
        \end{align*}
    \end{minipage}
    \text{\hspace{0.5cm}}
    \begin{minipage}{0.25\textwidth}
        \centering
        \includegraphics[width=0.9\textwidth]{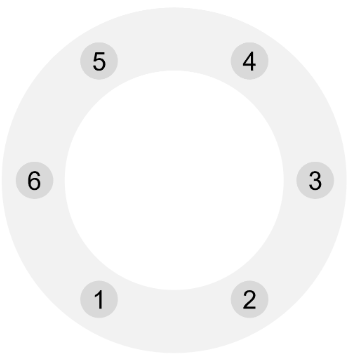} 
        \begin{align*}
        \text{(f)}    
        \end{align*}
    \end{minipage}
\caption{\textit{Reflection and rotation symmetry.} The qubit indexing for a circuit mapping onto a closed loop of six qubits, with (a) the original indexing, (b) the first rotation symmetry indexing, (c) the second rotation symmetry indexing, (d) the first reflection and rotation symmetry indexing, (e) the second reflection and rotation symmetry indexing, and (f) the third reflection and rotation symmetry indexing.}
\label{fig4}
\end{figure*}

It is also possible to combine the reflection and rotation symmetry reductions. For $n$-gon device topologies, the group of symmetries is defined by the dihedral group. The dihedral group for an $n$-gon has order 2$n$, this includes both rotation and reflection symmetries. A reflection and rotation symmetric circuit, like the circuit in Fig. \ref{fig2} (b), mapped onto a loop of qubits can exploit both both types of symmetry for error complexity reduction. This ensures the mappings for these two types of symmetry onto a closed loop of qubits on the device topology do not change the overall error channel. Fig. \ref{fig4} shows the different mappings possible for rotation and reflection symmetry on a closed loop of six qubits. Reflection and rotation symmetric mapping onto this closed loop means including all of the rotation symmetries for the original clockwise circuit orientation shown in Fig. \ref{fig4} (a)-(c), then reflecting the circuit and including the rotations for the anticlockwise orientation shown in Fig. \ref{fig4} (d)-(f).

It is again assumed that non-overlapping symmetric two qubit gates are used in the entangling gate layer. For the different instances of combined reflection and rotation parallelisation, an $n$ qubit circuit is run on a closed $n$ qubit subset loop with the circuit configurations denoted by the set of $n$-tuples $(2k + (-1)^{j}\cdot 1 \text{ mod } n+1,\text{ } 2k + (-1)^{j}\cdot2 \text{ mod } n+1,\text{ } \ldots,\text{ } 2k+(-1)^{j}\cdot(n-1) \text{ mod } n+1,\text{ } 2k+(-1)^{j}\cdot(n) \text{ mod } n+1)$ for $k\in\{0,\text{ }1,\text{ }\ldots,\text{ }\frac{n}{2}-1\}$ and $j=\{0,\text{ }1\}$. For each circuit orientation there are $\frac{n}{2}$ different possible mappings given by the different $n$-tuples. The qubit labelling is permutated according to $\pi^i$ and $\pi^i \circ \sigma$ for $i\in\{0,1,\ldots,\frac{n}{2}-1\}$. For all circuit mappings the same stochastic Pauli channel acts on the loop of qubits on the topology regardless of how the circuits are configured. As the two qubit gate positions remain unchanged by the reflection and the rotation symmetric mappings. The effective error channel is then
\begin{equation}
\begin{split}
    \mathcal{E}^{\text{eff}}(\rho) &=\frac{1}{n}\sum_{P\in \mathbf{P}^{\otimes n}} \bigg(\sum_{i=0}^{\frac{n}{2}-1} c_{\pi^i(P)} + \sum_{i=0}^{\frac{n}{2}-1} c_{\pi^i \circ \sigma(P)} \bigg) P \rho P^{\dagger}.\\
\end{split}
\end{equation}

\begin{theorem}
Reflection and rotation parallelisation results in a linear error complexity reduction, as the number of Pauli coefficients needed to describe the effective stochastic Pauli channel is 

\begin{equation}
\begin{split}
\normalfont|\{\textbf{c}^{ref,rot}\}| = \frac{1}{n}(2^{2n}+2^n +10n -20),
\end{split}
\end{equation}
where $n=4k$, $k\geq2$ and $k\in \mathbb{Z}$.
\end{theorem}

\begin{proof}
The result follows from combining the coefficient terms from theorems 1 and 2, giving 
\begin{equation}
\begin{split}
|\{\textbf{c}^{ref,rot}\}| &= \frac{2}{n}(2^{n-1}(2^{n}+1) - 6 - 3 - 1)\\
&+ 6 + 3 + 1\\
&= \frac{2}{n}(2^{n-1}(2^{n}+1) - 10) + 10\\
\end{split}
\end{equation}
which is then rearranged. The effective error channel resulting from the parallelisation has Pauli coefficients of the form: $c_{P'}=\sum_{i=0}^{\frac{n}{2}-1} c_{\pi^i(P)} + \sum_{i=0}^{\frac{n}{2}-1} c_{\pi^i \circ \sigma(P)}$, where $P\in \mathbf{P}^{\otimes n}$ and $P'\in \{\pi^i(P)\}_i\cup \{\pi^i \circ \sigma(P)\}_i$. The condition $k\geq2$ ensures that there the reflection and rotation reductions do not overlap and so may be considered independently. 
No reduction occurs for the coefficients of Pauli operators that are rotation and reflection invariant, that is the operators for which it is true that $\pi^i \circ\sigma(P)=P$ for all $i\in\{0,1,\ldots,\frac{2}{n}-1\}$. 
That is the 0 qubit subset for which there is 1 coefficient. The $n$ qubit subset, for which there are 6 rather than 9 invariant operators because the asymmetric two qubit operator pairs XY, YX, XZ, ZX, YZ and ZY after reflection are reduced to just XY, YZ and ZX. \textit{This notation refers to operators applied on all $\frac{n}{2}$ pairs of qubits on the loop, for example XZ denotes $(X\otimes Z)^{\otimes \frac{n}{2}}$}. So that the remaining symmetry transformation invariant operators acting over all pairs of qubits are XX, YY, ZZ, XY, YZ and ZX. The $n$/2 qubit subset has 3 coefficients, one for each type of local Pauli operator. As for each of these there is a single type of Pauli operator acting on the first qubit and identity acting on the second qubit for all qubit pairs in the $n$ qubit loop.

\end{proof}

\subsubsection{Permutation parallelisation}

If a device topology is a complete graph then all possible permutations of qubit assignments can be used for circuit mapping. For permutation parallelisation symmetry, the set of qubit labellings is given by the symmetry group Sym$(n)$. An example of a six qubit complete graph topology is shown in Fig. \ref{fig5}. The practical difficulty with implementing permutation parallelisation is in finding circuits with permutation symmetry. This is not possible for the hardware-friendly non-overlapping symmetric two qubit gate layers previously used. For permutation symmetric circuits an entangling gate layer that is invariant under the permutation symmetry transformations is required. 
It is assumed that we can apply a single type of symmetric two qubit gates, like the CZ gate, simultaneously between all interqubit connections on the complete graph topology. 
If this is the only multi-qubit gate layer used in the circuit, the circuit structure mapped onto the topology will be invariant under any permutation of qubit indexing. 
Likewise, for controlled quantum gates that have multiple control qubits, if this is the only multiqubit gate acting on these qubits in the circuit, then these control qubits would be permutationally symmetric. 
The permutation parallelisation effective error channel acts on $\rho$ as
\begin{equation}
\begin{split}
    \mathcal{E}^{\text{eff}}(\rho) &=\frac{1}{|\text{Sym}(P)|}\sum_{P\in \mathbf{P}^{\otimes n}}\sum_{S \in \text{Sym}(P)} c_{S(P)} P \rho P^{\dagger}.\\
\end{split}
\end{equation}

\begin{theorem}
Permutation parallisation results in an exponential error complexity reduction, as the number of Pauli coefficients needed to describe the effective stochastic Pauli error channel is
\begin{equation}
\begin{split}
\normalfont|\{\textbf{c}^{per}\}| 
&= \frac{1}{6}(n+1)(n+2)(n+3) \\
\end{split}
\end{equation}

\end{theorem}
\begin{proof} 
For a permutation-symmetric $n$ qubit circuit, all possible permutations of qubit labelling for the circuit mappings are used, denoted by the symmetric group of the qubit labellings Sym($n$). The circuit is mapped onto the qubits using each of the different mappings and sampled from, and their outputs combined with equal weight. 
All qubit orderings, or positioning of the local Pauli operators within each $n$ qubit Pauli operator, in the error channel are equivalent under permutation symmetry. Pauli operators which have the same numbers of each type of Pauli operator share a single coefficient, which is the average of the coefficients of these operators. The permutation-symmetric effective channel coefficients are all of the form: $c_{P'}=\frac{1}{|\text{Sym}(P)|}\sum_{S \in \text{Sym}(P)} c_{S(P)}$, where $P\in \mathbf{P}^{\otimes n}$ and $P'\in\text{Sym}(P)$. The number of distinct coefficients is equal to the number of $n$ element multisets of local Pauli operators, where each multiset has a different number of each type of local Pauli operator. And as there are four types of operator and n qubits, the number of coefficients for the channel is $n+4-1 \choose n$. Which is then expanded and rearranged for the result. 
\end{proof}

The main drawback of scaling up the symmetry reduction approach generally is that the reduction in error complexity is proportional to the number of different circuit mappings required to generate the symmetric effective channel. Meaning that for a polynomial reduction in complexity a polynomial number of different parallel circuits are required, and for an exponential reduction, exponentially many. However, if some randomly chosen subset of the total number of symmetry mappings is used then the effective error channel will approximate the symmetric effective error channel. With the closeness of the approximation determined by the size of the subset. Permutation parallelisation involves a super-exponentially scaling number of circuit mappings. If some randomly selected subset of the permutation symmetry group is used for parallelisation then the effective error channel will converge on the permutation-symmetric channel in terms of the size of the subset of symmetry circuit mappings used. This follows from applying Hoeffding's inequality to the convergence of the coefficients of the effective error channel with increasing symmetry mapping subset size \cite{hoeffding1963probability}. As each Pauli coefficient of the effective error channel, $c_P$, will converge on the corresponding coefficient of the permutation-symmetric channel, $c_{S(P)}$, as the number of permutations of the error channel included in the effective error channel is increased. To show this very similar reasoning can be applied as is used in the next section to bound the convergence of the randomised effective error channel.

One example of multiqubit gates with high permutational symmetry are controlled gates which have multiple control qubits. The local gate $G$ conditionally applied on the $n^{\text{th}}$ qubit depending on $n-1$ control qubits is denoted $(C)^{n-1}G$. One example of this type of gate for three qubits is the quantum CCNOT or Toffoli gate. For a gate of this kind one can permute the wires of the control qubits without changing the action of the gate. This could be used to produce a permutation parallelisation error complexity reduction for the part of the error channel acting on the $n-1$ control qubits. 

\begin{figure}
  \text{\hspace{4.7em}}\includegraphics[width=0.8\linewidth]{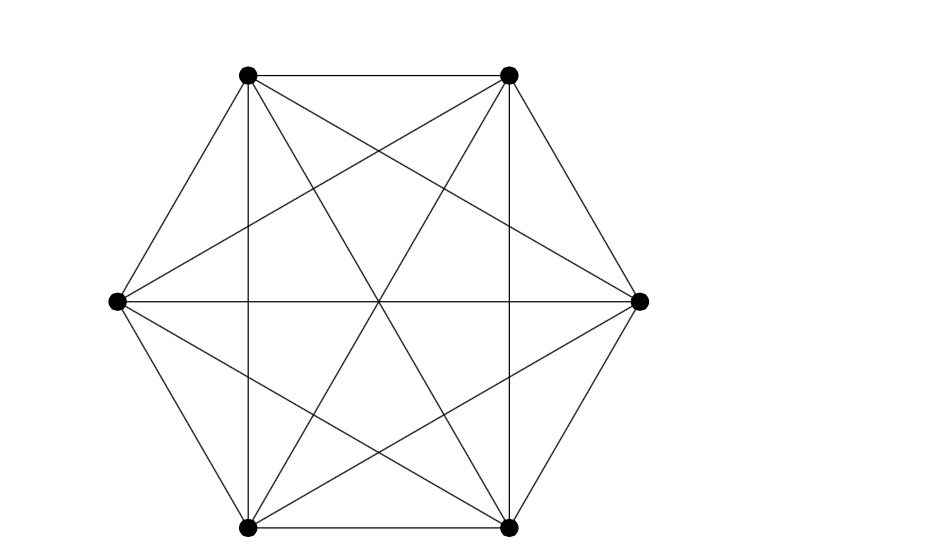}
\caption{\textit{Permutation symmetry.} A six qubit device topology consisting of a complete graph, where the vertices are qubits and the edges are interqubit connections. If all of the entangling gate layers of a circuit act symmetrically between all of the interqubit connections, the Pauli channel remains invariant under permutation symmetry transformations of the circuit.}
\label{fig5}
\end{figure}

\subsection{Randomisation \label{rand_sect}}

Now, instead of error symmetrisation, it is shown that randomisation can be applied to reduce error complexity. Through parallelising circuits across subsets of qubits with different stochastic Pauli channels, a randomised effective channel of lower relative complexity can be created. For a single qubit circuit parallelised across two qubits, two separate single qubit stochastic Pauli channels $\mathcal{E}_1$ and $\mathcal{E}_2$ act on two copies of the state $\rho$. If one combines the output state from the two different channels acting on $\rho$ with equal weight, then the two stochastic Pauli channels are averaged giving an effective error channel of

\begin{equation}
\begin{split}
\mathcal{E}(\rho)&= 2^{-1}(\mathcal{E}_1(\rho)+\mathcal{E}_2(\rho)) \\
&= 2^{-1}(c_{1,I}\rho + c_{1,X} X\rho X +c_{1,Y} Y\rho Y + c_{1,Z} Z \rho Z \\
&+ c_{2,I}\rho + c_{2,X} X\rho X +c_{2,Y} Y\rho Y + c_{2,Z} Z \rho Z)\\
&= 2^{-1}((c_{1,I}+c_{2,I})\rho + (c_{1,X}+c_{2,X}) X\rho X \\
&+ (c_{1,Y}+c_{2,Y}) Y\rho Y + (c_{1,Z}+c_{2,Z}) Z \\ 
&= c'_{I}\rho + c'_{X} X\rho X +c'_{Y} Y\rho Y + c'_{Z} Z \rho Z\\
&= \mathcal{E}'(\rho).
\end{split}
\end{equation}
This new effective Pauli channel $\mathcal{E}'$ has Pauli coefficients $\{c'_{I},c'_{X},c'_{Y},c'_{Z}\}$ which constitute the average of the coefficients for each of the respective Pauli operators from the two channels. This can be generalised for $N$ parallel circuits. Rather than parallelising using two single qubit circuits, $N$ single qubit circuits are used. Each of the circuits has a separate stochastic Pauli channel, $\{\{c_{i,I}, c_{i,X},c_{i,Y},c_{i,Z}\}\}_{i=1,\ldots,N}$, describing an independent probability distribution over the Pauli operators. The averaged error channel from the $N$ circuit parallelisation is 

\begin{equation}
\begin{split}
\mathcal{E}(\rho)&= N^{-1}\sum_{i=1,...N} \mathcal{E}_N(\rho) \\
&= N^{-1}(\sum_{i=1,...N}c_{i,I}\rho + \sum_{i=1,...N}c_{i,X} X\rho X \\
&+ \sum_{i=1,...N}c_{i,Y} Y\rho Y + \sum_{i=1,...N}c_{i,Z} Z)\\ 
&= c'_{I}\rho + c'_{X} X\rho X +c'_{Y} Y\rho Y + c'_{Z} Z \rho Z. \\
\end{split}
\end{equation}
Which can then be straightforwardly generalised to circuits of $n$ qubits and $n$ qubit stochastic Pauli channels. Making certain assumptions on the distributions of coefficients for the Pauli operators across the different parallel circuits, error randomisation can be used to show convergence on a simpler effective error channel. 
If the errors are time-dependent, the randomisation parallelisation could be performed using only a single set of qubits, with the circuit parallelisation performed in time rather than in space.

\subsubsection{Randomisation parallelisation using `nice' error model}

We assume an error model, referred to as (r,1), where the operator coefficients for the stochastic Pauli channels of the different parallel circuits are drawn from the same distribution. Parallelisation randomises the coefficients of the effective error channel, resulting in convergence on a single value as the number of parallel instances of the circuit is increased. In other words, the effective error channel converges on a global depolarizing channel. The probability distribution of coefficients over $N$ parallel circuits for a given Pauli operator $P$, where $P\in \mathbf{P}^{\otimes n}\backslash I$, is denoted $\{c_{i,P}\}_{i=1,\ldots,N}$. 
It is assumed that for the different parallel circuits the coefficients take values according to $c_{i,P} \sim \mathcal{D}(\eta, [0,1])$. Where $\mathcal{D}(\eta, [0,1])$ is any continuous distribution which takes values on the closed interval $[0,1]$ with mean $\eta$. This means that $\{c_{i,P}\}_i$ is a set of independent identically distributed random variables with a mean of $\eta$, and support on the bounded interval [0,1]. Importantly, the overall distribution of the Pauli coefficients, including the identity coefficient, for the error channel of each parallel circuit must still be normalised. I.e. $\sum_P c_{i,P}=1$.

\begin{theorem}
Randomisation parallelisation with the (r,1) error model results in an effective stochastic Pauli channel which is convergent upon a channel with 
\begin{equation}
\normalfont|\textbf{c}_{r,1}|=1.
\end{equation}
Pauli coefficients, 
 which is a global depolarizing channel with depolarizing parameter $\lambda=\frac{\eta(|\{\mathbf{P}^{\otimes n}\}|-1)}{|\{\mathbf{P}^{\otimes n}\}|}$.
\end{theorem}

\begin{proof}
If a circuit is parallelised over $N$ subsets of qubits, the coefficients of the effective error channel are $\{c'_{P}\}_{P\in \mathbf{P}^{\otimes n}}=\{N^{-1}\sum_{i=1,...N}c_{i,P}\}_{P\in \mathbf{P}^{\otimes n}}$. The coefficients relating to each Pauli operator are distributed across the original error channels of the different parallel circuits according to the previously stated assumptions of the (r,1) model. Hoeffding's inequality can then be used to upper bound the convergence of the effective channel coefficients on the value $\eta$ with increasing N. For a given Pauli operator $P\in \mathbf{P}^{\otimes n}$, this gives 
\begin{equation}
\begin{split}
&\text{Pr}\bigg[\bigg(\frac{1}{N}\sum_{i=1,...N}c_{i,P} -\eta \bigg)>\epsilon \bigg]\leq 2e^{-2N\epsilon^2} .\\
\end{split}
\end{equation}
The number of parallel circuits needed for a given accuracy and precision is $N\geq \frac{1}{2 \epsilon^2}\ln\frac{2}{\delta}$, where $\delta \geq 2e^{-2N\epsilon^2}$. That is, to get an ($1-\delta$)-confidence interval $\eta \pm \epsilon$ at least $\frac{1}{2 \epsilon^2}\ln\frac{2}{\delta}$ different error channels must be averaged over in the effective error channel. 
The coefficients of the effective channel all converge on a single value, $\eta$, as a function of the number of parallel circuits, $N$. Therefore the effective Pauli channel converges on a global depolarizing channel. 
The equivalence between the randomised effective error channel, with all Pauli coefficients equal to $\eta$, and a global depolarizing channel may be seen by regrouping the terms of the effective Pauli error channel
\begin{equation}
\begin{split}
\mathcal{E}(\rho)&= (1-\eta)\rho + \frac{\eta}{|\{\mathbf{P}^{\otimes n}\}|}\sum_{P\in \mathbf{P}^{\otimes n}} P \rho P^{\dagger} \\
&= \bigg(1- \frac{\eta(|\{\mathbf{P}^{\otimes n}\}|-1)}{|\{\mathbf{P}^{\otimes n}\}|}\bigg)\rho \\
& \hspace{6em}+ \frac{\eta}{|\{\mathbf{P}^{\otimes n}\}|}\bigg(\rho+\sum_{P\in \mathbf{P}^{\otimes n}} P \rho P^{\dagger}\bigg) \\
&= \bigg(1- \frac{\eta(|\{\mathbf{P}^{\otimes n}\}|-1)}{|\{\mathbf{P}^{\otimes n}\}|}\bigg)\rho + \frac{\eta(|\{\mathbf{P}^{\otimes n}\}|-1)}{|\{\mathbf{P}^{\otimes n}\}|}\frac{I}{n}\\
&=  (1 - \lambda)\rho + \lambda\frac{I}{n}.\\
\end{split}
\end{equation}
Hence the relation between the stochastic Pauli channel coefficient, $\eta$, and the depolarizing parameter, $\lambda$, is $\lambda=\frac{\eta(|\{\mathbf{P}^{\otimes n}\}|-1)}{|\{\mathbf{P}^{\otimes n}\}|}$. 
\end{proof}

For this error model there is no differentiation between the probabilities of higher and lower weight Pauli errors, as all are identically distributed. A potentially more realistic error model could be one in which Pauli operators of different weight can have different distributions over the coefficients. For example, it might be expected that lower weight errors occur with greater probability than higher weight ones. More specifically, that the weight 1 errors, that is to say local errors or errors that can be expressed in unordered tensor product form as $P_1\otimes I^{\otimes n-1}$ with $P_1\in\{X,Y,Z\}$, occur with the greatest probabilities. The next most probable type of error being the weight 2 operators, $P_1\otimes P_2 \otimes I^{\otimes n-2}$, and so on. 
In the following section a noise model is assumed that allows for variation in coefficient distributions of the Pauli operators with weight and qubit position across the different parallel circuits.

\subsubsection{Randomisation parallelisation using physically motivated error model}

Now a physically motivated error model is assumed, referred to as (r,2), allowing for variation in Pauli operator coefficient distributions across parallel circuits with the weight of the operators. 
In this error model, the coefficient distributions for the Pauli error operators over the different parallel circuits depend only on the relative positioning of the qubits acted on non-trivially by the operators, but not on operator type. 
Consider the component of an $n$ qubit stochastic Pauli channel $\mathcal{E}^P$ acting non-trivially on a specific subset of $m$ qubits where $m \leq n$. This component of the channel acts on the $m$ qubit subset with the set of $3^m$ Pauli operators denoted $\mathbf{G}= \{X,Y,Z\}^{\otimes m}$, with the identity operator acting on all $n-\nolinebreak m$ other qubits. The $m$ qubit subset this component of the channel acts on is labelled by the $m$-tuple $q=(q_1,\ldots,q_m)$, where $0\leq q_1<\ldots<q_m \leq n$. The component of the stochastic Pauli channel acting on the $m$ qubit subset $q$ of the $i^{th}$ parallel circuit is denoted 
\begin{equation}
\mathcal{E}^\mathbf{G}_{i,q}(\rho)=\sum_{G\in \mathbf{G}} c^{G}_{i,q} G \rho G^{\dagger}.
\end{equation}
The set of coefficients $\{c^{G}_{i,q}\}_{G\in \mathbf{G}}$ is the probability distribution over the set of Pauli operators $\mathbf{G}$, which act non-trivially only on the set of qubits denoted by the m-tuple $q$ for the $i^{\text{th}}$ parallel circuit. 

For any subset of qubits, $q$, it is assumed that the coefficients of operators which act non-trivially on the subset across the different parallel circuits take values according to $c^{\mathbf{G}}_{i,q} \sim \mathcal{D}(\eta_q, [0,1])$. Where $\mathcal{D}(\eta, [0,1])$ is any continuous distribution which takes values on the closed interval $[0,1]$ with mean $\eta_q$, with $0 \leq \eta_q \leq \frac{1}{2}$ and $i \in {1,\ldots,N}$. 
So that the average of the coefficients over the parallel circuits for each of these operators converge on $\eta_q$ with increasing $N$. 
Leading to an effective error channel consisting of a mixture of depolarizing channels, with a different depolarizing channel acting each different subset of $n$ qubits. Randomising an $n$ qubit stochastic Pauli channels by parallelisation gives an effective error channel of a probabilistic mixture of $2^n-1$ depolarizing channels, each acting on a different subset of qubits. 

\begin{theorem}
Randomisation parallelisation with the (r,2) error model results in an effective stochastic Pauli channel which is convergent upon a channel with 
\begin{equation}
\normalfont|\{\textbf{c}_{r,2}\}|=2^n - 1
\end{equation}
Pauli coefficients. 
Each coefficient relates to a depolarizing channel acting on a subset of qubits $q_j$ for $j\in\{0,\text{ }\ldots,\text{ }2^n-2\}$. And the relation between Pauli coefficient $\eta_{q_j}$ and depolarizing parameter $\lambda_{q_j}$ is $\lambda_{q_j}=\nolinebreak\frac{\eta_{q_j}(|\{\mathbf{P}^{\otimes n}\}|-1)}{|\{\mathbf{P}^{\otimes n}\}|}$.
\end{theorem}

\begin{proof}
The component of the effective error channel $\mathcal{E}'$ which acts non-trivially on an $m$ qubit subset $q_j$ is denoted $\mathcal{E}^{\mathbf{G}_j}_{q_j}$, where $0\leq m \leq n$, $j\in\{0,\text{ }\ldots,\text{ }2^n-2\}$, $\mathbf{G}_j= \{X,Y,Z\}^{\otimes m}$, and the channel coefficients are denoted $\{c^{G}_{q_j}\}_{G\in\mathbf{G}_j}$. For the (r,2) error model, the shared expected value for the distributions of coefficients of all of the Pauli operators which act non-trivially on the subset of qubits $q_j$ is $\eta_{q_j}$. Using Hoeffding's inequality, each coefficient of the effective error channel 
is, with high probability, $\epsilon$-close to the expected value $\eta_{q_j}$ in terms of $N$. If the number of parallel circuits, $N$, obeys the relation $N\geq \frac{1}{2 \epsilon^2}\ln\frac{2}{\delta}$. Then for the average of the set of coefficients $\{c^{G}_{i,q_j}\}_{i\in\{1,\ldots,N\}}$ relating to the Pauli operator $G$, i.e. the coefficient of $G$ in the effective error channel $c^{G}_{q_j} = \sum_{i=1}^N c^{G}_{i,q_j}/N$, it is true that $\text{Pr}[\sum_{i=1}^N (c^{G}_{i,q}/N) - \eta_{q_j} >\epsilon]\leq 2e^{-2N\epsilon^2}\leq \delta$. 
So the randomised effective error channel acting on the qubit subset $q_j$ is well approximated by
\begin{equation}
\begin{split}
\mathcal{E}^{\mathbf{G}_j}_{q_j}(\rho)&=\sum_{G\in \mathbf{G}_j} c^{G}_{q_j} G \rho G^{\dagger}\\
&\approx\sum_{G\in \mathbf{G}_j} \eta_{q_j} G \rho G^{\dagger}.\\
\end{split}
\end{equation}
With the result that the reduction in Pauli coefficients for the component of the effective error acting on this $m$ qubit subset is from $3^m$ to 1. This is true for any component of the effective error channel acting on a restricted subset of qubits. The effective error channel $\mathcal{E}^{\text{eff}}$ is then a convex combination of error channels, each of which acts on a different subset of the set of $n$ qubits. So that
\begin{equation}
\begin{split}
\mathcal{E}^{\text{eff}}(\rho)&=\sum_{j=1}^{2^n -1}\mathcal{E}^{\mathbf{G}_j}_{i,q_j}(\rho)\\
&\approx\sum_{j=1}^{2^n -1}\sum_{G\in \mathbf{G}_j} \eta_{q_j} G \rho G^{\dagger}.\\
\end{split}
\end{equation}
This effective error channel consists of a combination of $2^n -1$ depolarizing channels. Each of these has a single depolarizing parameter and so the effective channel therefore has $2^n -1$ coefficients. With the expression for the depolarizing channel parameters in terms of Pauli coefficients derived in a similar way as previously, but with a different depolarizing parameter for each subset of qubits.
\end{proof}

An illustrative example of randomisation parallelisation for a two qubit subset of a three qubit circuit is given in the Appendix \ref{appendix}.

\subsection{Symmetric Randomisation}

It is possible to combine the symmetry and randomisation error complexity reductions. Only the (r,2)
error model is considered since no further reduction is possible for the (r,1) error model after randomisation parallelisation. For the (r,2) error model, randomisation parallelisation creates an effective error channel in which coefficients of Pauli operators that act non-trivially on the same subset of qubits approach a single shared value. So that the randomised effective channel consists of a mixture of depolarizing channels, each acting on a different subset of the set of $n$ qubits. Whereas symmetry parallelisation generates an effective channel where the coefficients of Pauli operators that can be mapped onto each other by the relevant symmetry transformation are equal. Consequently, employing these two approaches together causes two different reductions to take place. 

\begin{subtheorem}   
\label{subth1}   
Reflection and randomisation parallelisation using the ($r,2$) error model results in an effective stochastic Pauli channel which is convergent upon a channel with 
\begin{equation}
\begin{split}
\normalfont|\{\textbf{c}_{r,2}^{ref}\}| &= \frac{1}{2}\bigg(\sum_{i=0}^n {n \choose i} + \sum_{j=0}^{n/2} {n/2 \choose j}\bigg)
\end{split}
\end{equation}
Pauli coefficients.
\end{subtheorem}

\begin{proof}
The reflection symmetry reduction leads to Pauli operators which can be mapped to each other by the reflection symmetry transformation sharing by a single coefficient in the effective error channel. 
So that reflection parallelisation results in a total number of coefficients given by
$$
|\{\textbf{c}^{ref}\}| =\frac{1}{2}\bigg(\sum_{i=0}^n {n \choose i}3^i + \sum_{j=0}^{n/2} {n/2 \choose j}3^j\bigg), \\
$$
Applying randomisation along with reflection parallelisation means that the coefficients of Pauli operators which act non-trivially on the same subset of qubits, and on the reflection of this subset, converge on the same value. For each subset of $i$ qubits from the set of $n$ qubits, according to Thm. 6 the distinct coefficients for operators acting non-trivially on this subset converge from $3^i$ coefficients to 1 coefficient. So that the power of 3 terms in the previous expression all converge to 1 with $N$, and the result follows. 
\end{proof}

\begin{subtheorem}   
    \label{subth2}   
Rotation and randomisation parallelisation using the ($r,2$) error model results in an effective stochastic Pauli channel which is convergent upon a channel with 
\begin{equation}
\begin{split}
\normalfont|\{\textbf{c}_{r,2}^{rot}\}| 
 &=  \frac{2}{n}\sum_{i=0}^n {n \choose i} + \frac{4n-8}{n}\\
    \end{split}
\end{equation}
Pauli coefficients.
\end{subtheorem}

\begin{proof}
The result follows in a similar way as for Sub-Thm. \ref{subth1}.
\end{proof}

\begin{subtheorem} 
\label{subth3}   
Reflection, rotation and randomisation parallelisation using the ($r,2$) error model results in an effective stochastic Pauli channel which is convergent upon a channel with 
\begin{equation}
\begin{split}
 \normalfont   |\{\textbf{c}_{r,2}^{ref,rot}\}| 
&= \frac{1}{n}\bigg(\sum_{i=0}^n {n \choose i} + \sum_{j=0}^{n/2} {n/2 \choose j}\bigg) + \frac{4n-8}{n}\\
\normalfont\end{split}
\end{equation}
Pauli coefficients.
\end{subtheorem}

\begin{proof}
The result follows in a similar way as for Sub-Thm. \ref{subth1}.
\end{proof}

\begin{subtheorem} 
\label{subth4}   
Permutation and randomisation parallelisation using the ($r,2$) error model results in an effective stochastic Pauli channel which is convergent upon a channel with 
\begin{equation}
\begin{split}
\normalfont|\{\textbf{c}_{r,2}^{per}\}| &= n+1 \\
\end{split}
\end{equation}
Pauli coefficients.
\end{subtheorem}

\begin{proof}
The result follows in a similar way as for Sub-Thm. \ref{subth1}.
\end{proof}

\section{\label{mult}  Parallelisation for multiple error channels}

Instead of considering the case of a single error channel acting on each parallel circuit, we now consider circuits with multiple error channels. 
In the following, noisy circuits are modelled as a concatenation of ideal gate layers immediately followed by their associated error channels. The state of the system after the noisy implementation of a circuit with $m$ layers of gates is then of the form

\begin{equation}
\begin{split}
    \rho_{out}&=\mathcal{E}_m U_m\ldots\mathcal{E}_2U_2\mathcal{E}_1U_1(\rho).\\
\end{split}
\end{equation}
Where $U_k$ is the $k$-th layer of ideal gates, and $\mathcal{E}_k$ the error channel associated with it. 
The symmetry reductions are not directly extendable to multiple error channels unless extra conditions are imposed on the input circuit. This is because the errors would need to preserve the relevant symmetry during propagation for the symmetry reductions to be valid. 
This is also true of the symmetric randomisation reductions. These conditions are stated and their implications briefly discussed. 
As the distributions of the Pauli coefficients for the different error channels within each circuit are assumed to be independent, the randomisation reductions extend directly to the multiple channel case. 

\subsection*{Symmetry}

To extend the symmetry reductions to circuits with multiple error channels it is required that the propagation of errors be symmetric for the results from the previous section to hold. 
For multiple error channels and without certain strict conditions on the gates, such as, for example, that all single qubit gates are local Pauli operators, the effective channel is in fact an average of error channels each acting on a different state. 
The extension of the reflection symmetry reduction to multiple error channels serves to illustrate this. The analysis for other types of symmetry is very similar. In the reflection reduction an error channel and its reflection are combined to create an effective channel of the form
\begin{equation}
\begin{split}
    \rho_{out}&=\frac{1}{2}(\mathcal{E}(\rho) +  \mathcal{E}^{R}(\rho))\\
    &=\frac{1}{2}(\mathcal{E} +  \mathcal{E}^{R})(\rho)\\
    &=\mathcal{E}^{\text{eff}}(\rho).\\
\end{split}
\end{equation}
Where $\mathcal{E}$ is the original error channel, $\mathcal{E}^{R}$ is its reflection and $\mathcal{E}^{\text{eff}}$ is the effective error channel. In the case where two error channels are instead applied this becomes
\begin{equation}
\begin{split}
    \rho_{out}&=\frac{1}{2}(\mathcal{E}_2\mathcal{E}_1(\rho) +  \mathcal{E}_2^R\mathcal{E}_1^R(\rho))\\
    &=\frac{1}{2}(\mathcal{E}'(\rho) +  \mathcal{E}'^{R}(\rho))\\
    &=\mathcal{E}^{\text{eff}}(\rho).\\
\end{split}
\end{equation}
Here it is again possible to combine the channels to create a simpler reflection symmetric effective channel. However, if arbitrary gate layers are included between each error channel, the output state instead becomes
\begin{equation}
\begin{split}
    \rho_{out}&=\mathcal{E}_2U_2\mathcal{E}_1U_1(\rho) +  \mathcal{E}_2^R U_2 \mathcal{E}_1^R U_1 (\rho)\\
    &=\mathcal{E}_2(\rho') +  \mathcal{E}_2^R(\rho'').\\
\end{split}
\end{equation}
The reflection symmetry needed for the reduction will be lost due to asymmetric error propagation, unless local gates are chosen so that the errors will propagate in such a way that the symmetry between the original and reflected error channels is preserved. If the local gates in the input circuit are all local Pauli operators, it is possible to propagate all of the Pauli error channels to the end of the circuit. With errors propagating symmetrically in the original and the reflected circuit. Due to the relation $(P_1\otimes\ldots\otimes P_n)(P_1'\otimes\ldots\otimes P_n')=(-1)^a(P_1'\otimes\ldots\otimes P_n')(P_1\otimes\ldots\otimes P_n)$ where $a=0$ if the Pauli operators commute and otherwise $a=1$. Then it is again possible to attain the reduction in the final combined error channel
\begin{equation}
\begin{split}
    \rho_{out}&=\mathcal{E}_2U_2\mathcal{E}_1U_1(\rho) +  \mathcal{E}_2^R U_2 \mathcal{E}_1^R U_1 (\rho)\\
    &=\mathcal{E}_2 \mathcal{E}'_{1} U_2U_1(\rho) +  \mathcal{E}_2^R \mathcal{E}_1'^{,R} U_2 U_1 (\rho)\\
    &=\mathcal{E}' U_2 U_1(\rho) +  \mathcal{E}'^{,R} U_2 U_1(\rho)\\
    &=\mathcal{E}^{\text{eff}}(\rho).\\
\end{split}
\end{equation}
Another alternative could instead be that local gate layers are required to be reflection symmetric. While errors would no longer be kept within the Pauli group during propagation, this would ensure an error propagation that preserves symmetry.

\subsection*{Randomisation}

Although error propagation is problematic in the extension of symmetric reductions, this issue does not feature in extending the randomisation reductions from single to multiple error channels. 
Since in the multiple channel version of randomisation parallelisation each effective channel is randomised independently, all of the effective channels converge in the same way as in the single channel case and the results hold. 
In other words, if the distributions of coefficients for the different error channels within each circuit are independent, then the error complexity reductions follow in the same way as in the single channel case. If each of the parallel circuits is sampled from with equal weight, the effective output state being sampled is
\begin{equation}
\begin{split}
    \rho_{out}&=\mathcal{E}^{\text{eff}}_m U_m\ldots\mathcal{E}^{\text{eff}}_2U_2\mathcal{E}^{\text{eff}}_1U_1(\rho).\\
\end{split}
\end{equation}
Where each error channel in the set $\{\mathcal{E}^{\text{eff}}_i\}_{i=1,\ldots,m}$ is the randomised effective error channel from the average over the error channels for each gate layer of the parallel circuits. And with the convergence behaviour of the effective error channels depending on the pertinent noise model.

\subsection*{Symmetric Randomisation}

Combining reflection symmetry parallelisation and randomisation parallelisation, the effective output state from sampling randomly with equal weight from the parallel circuits is
\begin{equation}
\begin{split}
    \rho_{out}&=\mathcal{E}^{\text{eff}}_m U_m\ldots\mathcal{E}^{\text{eff}}_2U_2\mathcal{E}^{\text{eff}}_1U_1(\rho)\\
     &\hspace{2em}+\mathcal{E}^{\text{eff},R}_m U_m\ldots\mathcal{E}^{\text{eff},R}_2U_2\mathcal{E}^{\text{eff},R}_1U_1(\rho).\\
\end{split}
\end{equation}
It is therefore the case that the same issue applies as in the instance of symmetry alone. Namely that the error channel propagation for the symmetry reductions prevents extension to multiple error channels with arbitrary gates. 
As the asymmetric error propagation prevents the coefficient matching necessary for the error complexity reduction. However, if the circuit conditions described previously regarding the extension of the symmetry reductions to multiple channels are met, then the symmetric randomisation results will also extend to multiple channels. The same arguments are also true for the other instances of symmetric randomisation.

\section{\label{app} Applications}

There are many ways these methods might be usefully applied; some examples of potential applications are briefly described in this section. The first is applying symmetry error simplification to reduce the sample cost of measurement error mitigation. The second is employing randomisation to enhance the effectiveness of noise-estimation circuit mitigation. And the final application is using randomisation to make noisy circuit performance more predictable and robust to time-dependent error.

\subsection*{Reducing sample complexity for quantum measurement error mitigation}

Measuring qubits is a large source of error in many types of quantum hardware. This has motivated the development of a wide array of measurement error mitigation techniques, designed to reduce the effects of measurement errors on the computational output through classical postprocessing \cite{chen_detector_2019,hamilton_scalable_2020,maciejewski_mitigation_2020,berg_model-free_2022}. Matrix-inversion measurement error mitigation is one of these techniques \cite{hamilton_scalable_2020,maciejewski_mitigation_2020,yang_efficient_2022}. It involves the generation of a stochastic transition matrix encoding information about the measurement noise. This matrix is experimentally derived then inverted and used to multiply the output vector of a noisy computation. It is assumed that the measurement noise can be modelled by a stochastic matrix, such that
\begin{equation}
    \label{eqnoisemap1}
    {p_{\text{noisy}}}=\mathcal{S}{p_{\text{ideal}}}.
\end{equation}
Where $p_{\text{noisy}}$ and ${p_{\text{ideal}}}$ are the output distribution vectors with and without measurement noise respectively. In the absence of further simplifying assumptions, such as restricting the weight of the correlated errors, these stochastic matrices are of size $2^n \times 2^n$. Matrix characterisation presents a significant difficulty for scaling up the technique because of the sample cost. 

\begin{figure}
  \text{\hspace{-0.9em}}\includegraphics[width=0.85\linewidth]{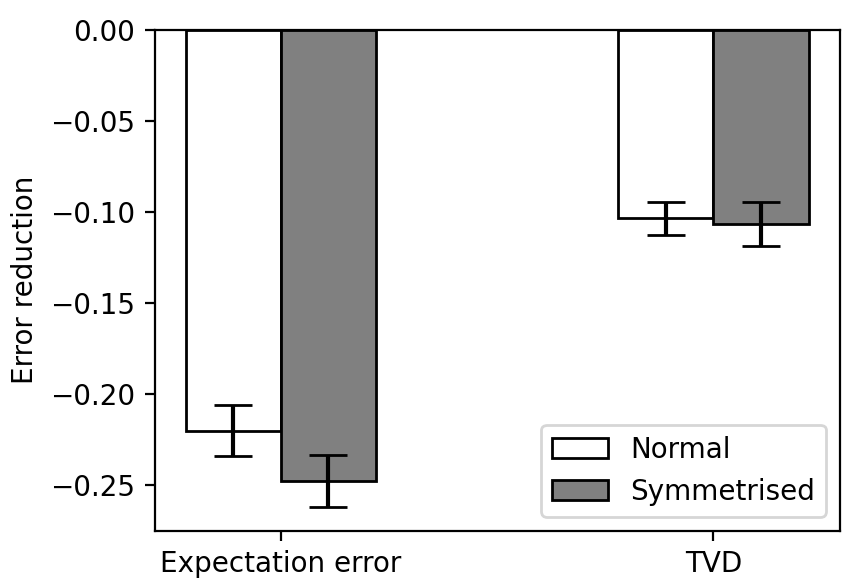}
\caption{\textit{Quantum measurement error mitigation with and without the symmetry reduction.} A plot of the relative performance of measurement error mitigation both with and without the rotation and reflection symmetry reduction, in terms of expectation value error and total variation distance. The mitigation effectiveness of the normal measurement error mitigation and the symmetrised version is seen to be very similar. For the mitigation without symmetry reduction $160000$ samples were used, as compared to the $60000$ needed for the mitigation with symmetry reduction.}
\label{fig7}
\end{figure}

To generate the entries of the stochastic measurement error matrix different computational basis states are prepared using $X(\pi)$ gates, and then are immediately measured in the computational basis. Error from imperfect measurement means that, rather than just the prepared basis state, a noisy output is obtained consisting of a distribution over the basis states. Only bit flip or Pauli X errors are registered, which can be viewed as a special case of a stochastic Pauli channel. Where for $n$ qubits the set of operators is $P_X \in\{I,X\}^{\otimes n}$ instead of $\{I,X,Y,Z\}^{\otimes n}$. The measurement error channel, $\mathcal{E}^M$, acting on the output state is then
\begin{equation}
\begin{split}
\mathcal{E}^M(\rho)&=\sum_{P_X\in \{I,X\}^{\otimes n}} c_{P_X}P_X \rho P_X^{\dagger}.\\
\end{split}
\end{equation}
The following experiments were run on the superconducting hardware to test whether symmetry parallelisation can provide a sample complexity improvement for measurement error mitigation without reducing mitigation effectiveness. Information on the hardware used is provided in the Appendix. A four qubit subset connected in a closed loop on the device topology was used in each experiment, so that there were $16$ basis states and the stochastic measurement error matrix was of size $16 \times 16$ entries. 
There are eight ways to map the circuit onto the closed four qubit loop using rotation and reflection symmetry. With the result that instead of having to characterise the output distributions for the set of 16 basis states, one need only characterise the distributions for 6 states: $\{\ket{0000}, \ket{0001},\ket{0011},\ket{0101},\ket{0111},\ket{1111}\}.\text{ }$
As for example the state $\ket{0011}$ can be mapped onto the states $\ket{0110}$, $\ket{1100}$ and $\ket{1001}$ by symmetry transformations. And so the rows in the transition matrix for these basis states contain the same information, which is the average of the noisy output distributions of the states.

In the experiment, a couple of two qubit Bell pairs were generated using the circuit shown in Figure \ref{fig9} (a). This circuit creates the state
\begin{equation}
\ket{\psi}=\frac{1}{\sqrt{2}}(\ket{00}+\ket{11})\otimes\frac{1}{\sqrt{2}}(\ket{00}+\ket{11}).
\end{equation}
Reflection and rotation parallelisation was used for the circuit mappings onto the qubits. Each mapping was used to obtain an empirical approximation of the noisy output probability distribution. 
In the first instance, mitigation is used to reduce the effects of readout errors for each of the circuit mappings. In the second instance, reflection and rotation parallelisation was used with the mitigation. The circuits were run using each of the eight different mappings, and the outputs were combined to create an average output distribution. This averaged output was then mitigated using the reflection and rotation symmetry parallelised transition matrix. In the normal measurement error mitigation, 160,000 samples were used, while 60,000 were used for the symmetry parallelised mitigation. As can be seen from the plot in Figure \ref{fig7}, approximately the same error reduction was achieved for the symmetry-aided mitigation as for the normal mitigation, but with 100,000 fewer samples. 

\subsection*{Enhancing effectiveness of noise-estimation circuit mitigation}

Error mitigation by noise-estimation circuits assumes that depolarizing noise is the dominant source of error in the input noisy circuit \cite{urbanek_mitigating_2021}. Previously it was shown that randomisation parallelisation can generate averaged noise that converges on depolarizing. In the following experiment, noise-estimation circuit mitigation is performed both with and without randomisation to test whether randomisation makes the technique's assumptions more valid, and causes an improvement in mitigation performance. In the technique, the output expectation value from sampling from an input circuit using a noisy quantum device is assumed to be of the form
\begin{equation}
O_{noisy} = \lambda O_{ideal}.
\end{equation}
Where $O_{noisy}$ and $O_{ideal}$ are the noisy and ideal expectation values respectively, and $\lambda$ is the depolarizing parameter. A simplified version of the input circuit is used to experimentally derive an approximation of the depolarizing parameter, denoted $\lambda'$. This is then used to postprocess the noisy expectation value output, inverting the action of the depolarizing noise on the expectation value derived from running the circuit on noisy hardware. If it is approximately true that
\begin{equation}
\frac{\lambda}{\lambda'} \approx 1,
\end{equation}
the mitigated expectation value should be close to the ideal, and the mitigation successful. 

In the experiments, error mitigation was first performed without the randomisation parallelisation, and subsequently with it. Four qubit circuits were used with four gate layers consisting of alternating CZ and local X rotation gate layers, with eight rotation parameters overall. The circuit structure is shown in Figure \ref{fig9} (b). The ideal circuit was randomly initialised and trained to output ideal expectation values from the set: $\{-1,-0.5,0.5,1\}$. The parameters obtained for each of these ideal expectation values were then used to run the circuit mapped onto three closed loops of the superconducting quantum device. When depolarizing factors less than $0.1$ were recorded, these were rejected as failed circuit runs. The noisy expectation value of the first qubit was mitigated both with and without randomisation parallelisation. The results are shown in the plot in Figure \ref{fig8}. Using randomisation parallelisation with the mitigation lead to an average improvement relative to the normal mitigation of $41\%$. Indicating randomisation parallelisation generated effective error channels that more closely approximated the depolarising noise assumed for the mitigation technique than was the case for the individual parallel circuits.

\begin{figure}
\text{\hspace{0em}}\includegraphics[width=0.85\linewidth]{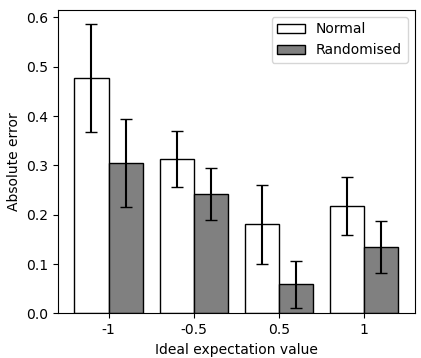}
\centering
\caption{\textit{Absolute error after noise-estimation circuit mitigation with and without randomisation parallelisation.} A plot of the relative performance of noise-estimation error mitigation both with and without randomisation parallelisation. The absolute error of the mitigated expectation value output from each of the parallel circuits is recorded for the set of ideal expectation values $\{-1,-0.5,0.5,1\}$, this is then used to calculate the average normal post-mitigation error. For the randomised error mitigation, the average of the noisy outputs is mitigated using the average depolarizing parameter. The main assumption of noise-estimation circuit mitigation is depolarizing noise. As the technique appears to work better when randomisation is used this implies that the technique assumptions are more valid with than without randomisation. And that the effective error channels from randomisation more closely approximate depolarizing noise than is the case for the error channels of the individual parallel circuits.
}

\label{fig8}
\end{figure}
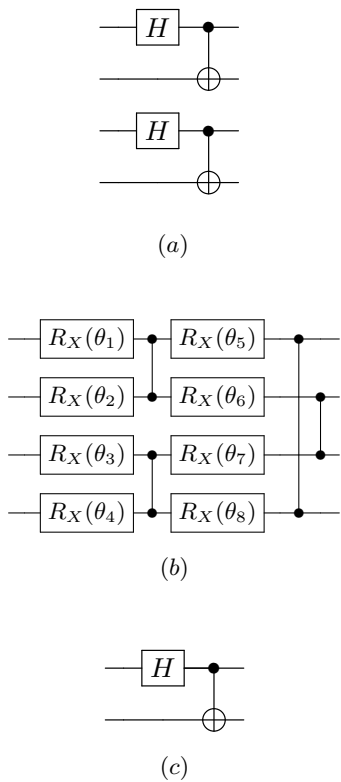
\begin{figure}
\centering
\begin{minipage}{0.4\textwidth}
\text{\hspace{0em}}
\begin{align*}\hspace{-0.7em}
\scalebox{1.15}{
\Qcircuit @C=0.65em @R=0.8em {
  &\qw &\gate{H} & \ctrl{1}  & \qw\\
 &\qw & \qw & \targ & \qw\\
 &\qw & \gate{H} & \ctrl{1}  & \qw\\
 &\qw & \qw & \targ  & \qw  \\
}
}
\end{align*}
\vspace{-0.7em}
\begin{align*}
(a)
\end{align*}
\label{fig:sub12}
\end{minipage}%
\vspace{0.6em}
\text{\hspace{1em}}
\newline
\begin{minipage}{0.4\textwidth}

\hspace{-3em}
\scalebox{1}{
\Qcircuit @C=0.65em @R=0.8em {
&\qw &\gate{R_{X}(\theta_1)} & \ctrl{1}     & \gate{R_{X}(\theta_5)} &\qw &\ctrl{3} & \qw & \qw \\
&\qw & \gate{R_{X}(\theta_2)} & \control \qw & \gate{R_{X}(\theta_6)} &\qw &\qw & \ctrl{1} &\qw \\
&\qw & \gate{R_{X}(\theta_3)} & \ctrl{1}     & \gate{R_{X}(\theta_7)} &\qw &\qw & \control \qw & \qw \\
&\qw & \gate{R_{X}(\theta_4)} & \control \qw & \gate{R_{X}(\theta_8)} &\qw & \control \qw & \qw  & \qw\\
}
}

\vspace{-0.4em}
\begin{align*}
\hspace{-1.75em}(b)
\end{align*}

\label{fig:sub22}
\end{minipage}

\vspace{1.6em}
\text{\hspace{1em}}
\newline
\begin{minipage}{0.4\textwidth}

\hspace{-3em}
\scalebox{1.15}{
\Qcircuit @C=0.65em @R=0.8em {
  &\qw &\gate{H} & \ctrl{1}  \qw & \qw\\
 &\qw & \qw & \targ & \qw\\
}
}

\vspace{-0.4em}
\begin{align*}
\hspace{-1.75em}(c)
\end{align*}

\label{fig:sub22}
\end{minipage}

\caption{\textit{Circuits used in the quantum error mitigation experiments.} Circuit diagrams illustrating the circuits used in the experiments for (a) measurement error mitigation, (b) noise-estimation circuit mitigation, and (c) stabilising time-dependent noise.}
\label{fig9}
\end{figure}

\subsection*{Improving noisy circuit performance predictability and robustness to time-dependent error}

\begin{figure}[!htb]
\text{\hspace{0em}}\includegraphics[width=0.99\linewidth]{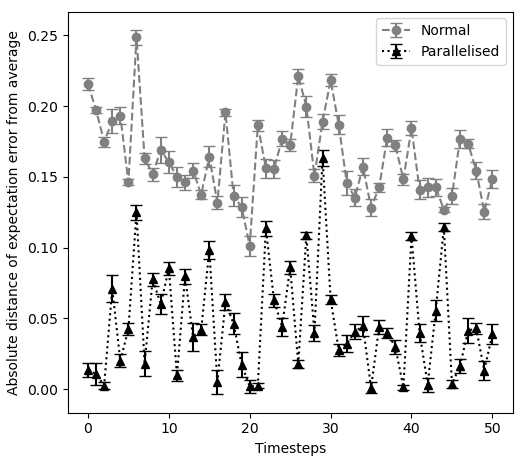}
\centering
\caption{\textit{Absolute difference of expectation value error from the time-series average error with and without parallelisation.} A plot of variation of the absolute distance of expectation value error from the time-series average with time from numerical simulation; both with and without parallelisation. Two qubit Bell states were prepared and measured in the computational basis. The effects of time-dependent noise were simulated by randomly generating error values from the uniform distribution $U(0,0.2)$, and using these as error parameters for a Pauli bit-flip noise channel. Ten two qubit parallel circuits were used in the experiment, which ran over the course of 50 timesteps.
}

\label{figtime}
\end{figure}

Metrics quantifying average error rates are often used in quantum device benchmarking and as indicators of noisy circuit performance. For example, the suite of techniques collectively referred to as randomized benchmarking are primarily focused on assessing quantum device performance by characterising average gate error rates \cite{knill_randomized_2008,helsen_general_2022,proctor_what_2017,proctor_measuring_2022}. 
There are a number of issues with using gate performance as a direct estimator of circuit performance. 
In particular, complex propagation and accumulation of gate errors can lead to unpredictable device performance when running computations. There are also benchmarking techniques like quantum volume that quantify how well random circuits of different widths and depths can be run \cite{cross_validating_2019}. The problem with metrics of this kind is the potential divergence between how noise affects the performance of random and structured circuits. 
Errors can constructively interfere during propagation, accumulating in a maximally damaging way and giving rise to worse than predicted circuit performance. Alternatively, errors can also destructively interfere, cancelling out and causing better than expected performance. Without exponentially costly simulation of the quantum noise behaviour, it is difficult to understand how errors are accumulating and their effect on the computational output. 

Total variation distance (TVD) is a commonly used distance measure for probability distributions \cite{nielsen_quantum_2010}. 
With $N$ parallel circuits, the set of TVDs for the noisy classical output distributions of the parallel circuits is the set $\{\text{TVD}(\text{P}(\rho_i), \text{P}(\rho_{\text{ideal}}))\}_{i=1}^{N}$. Where $\text{P}(\rho_i)$ is the classical probability distribution from measuring the noisy output state $\rho_i$, which depends on the error channels and propagation of the $i^{\text{th}}$ parallel circuit, and $\text{P}(\rho_{\text{ideal}}))$ is the ideal output distribution. The total variation distance of the effective output distribution from parallelisation is $\text{TVD}_{\text{eff}}=N^{-1}\sum^N_i \text{TVD}(\rho_i, \rho_{\text{ideal}})$, and should be bounded between the best and worst case TVDs, denoted $\text{TVD}_{\text{best}}$ and $\text{TVD}_{\text{worst}}$ respectively. So that
\begin{equation}
\text{TVD}_{\text{best}} \leq \text{TVD}_{\text{eff}} \leq \text{TVD}_{\text{worst}}.
\end{equation}
Parallelisation leads to different error propagation and accumulation behaviour for the parallel circuits being averaged over in the combined output. Strongly constructive or destructively interfering error propagation should, to some extent at least, average out. Also, the results in Fig. \ref{fig8} seem to indicate that randomisation causes the effective errors to more closely approximate depolarizing noise than the noise for the individual circuits. In theoretical analysis of noisy circuit behaviour, noise is often modelled as depolarizing due to the simplicity of its effect on the quantum state; it causes a decay towards the maximally mixed state with increasing circuit depth. If randomisation parallelisation indeed results in effective error channels that more closely approximate depolarizing channels, then it could help circuit performance to correspond more accurately with what is indicated by information provided by benchmarking techniques. 

Another factor that can reduce the accuracy of predictions of noisy circuit performance is time-dependent error. Parallelistion can be used render circuit performance more robust to the effects of time-dependent error. Numerical simulation of the effects of parallelisation on time-dependent errors indicate a consistent reduction in the variation of the parallelised error as compared to individual circuit error. The results of the experiment are shown in Fig. \ref{figtime}. The circuit shown in Fig. \ref{fig9} (c) was used in ten parallel circuits to produce the two qubit Bell state
\begin{equation}
\ket{\psi}=\frac{1}{\sqrt{2}}(\ket{00}+\ket{11}).
\end{equation}
A time-dependent bit-flip error channel was simulated for each of the circuits, with different error parameters randomly generated from the uniform distribution $U(0,0.2)$ at each timestep for each circuit. Parallelisation demonstrates reduced variation in time of the absolute difference of the expectation value error from the time-series average. Which indicates parallelisation's potential utility in stabilising the effects of time-dependent errors and in rendering noisy circuit performance more predictable.

\section{Summary and outlook}

We have presented methods to generate less complex error channels using quantum circuit parallelisation for error symmetrisation and randomisation. We refer to the action of simplifying errors by these methods as symmetry error complexity reduction and randomisation error complexity reduction, respectively. 
Symmetry error complexity reductions were detailed for the instances of reflection, rotation and permutation symmetry. 
While randomisation reductions were analysed in the context of two different error models. In the first, the resulting effective error channel was shown to be a global depolarizing channel. In the second, the effective error channel converged on a convex combination of depolarizing channels. 
The symmetrisation and randomisation methods were also combined to provide additional reductions. 
In the initial analysis, a single error channel was considered, and these results were then generalised to multiple error channels. We then proposed some practical applications for these methods. The use of symmetry error complexity reductions to lower the sample cost of measurement error mitigation. The implementation of randomisation reductions to improve the effectiveness of noise-estimation circuit mitigation. And the general application of these methods to improve the predictability of noisy circuit performance and robustness to time-dependent error. We discuss how this could improve circuit performance predictions based upon average gate performance metrics provided by noise characterisation techniques like randomized benchmarking \cite{magesan_efficient_2012, mckay_three-qubit_2019, helsen_general_2022-1, epstein_investigating_2014, magesan_scalable_2011, onorati_randomized_2019, magesan_characterizing_2012}. Experiments were run on superconducting hardware and numerically simulated to demonstrate these applications for selected test circuits.

There are many possible future research directions that would build upon this work. Firstly, experimental examinations of the validity of the error assumptions used and their applicability for different types of quantum hardware. This would test the degree to which the different assumptions used for the error complexity reductions are compatible with different types of hardware. This might be done by using a Pauli noise characterisation technique, like, for example, cycle benchmarking \cite{erhard_characterizing_2019, hashim_randomized_2021, chen_learnability_2023, yeter-aydeniz_measuring_2022, hashim_benchmarking_2022}. 
Likewise, experimental testing of the hypothesis that randomisation would improve the accuracy of circuit performance predictions based on techniques that provide information on average gate performance \cite{knill_randomized_2008,helsen_general_2022,proctor_what_2017,proctor_measuring_2022}. 
There is also great scope for investigation into further reductions possible with different device topologies and different types of symmetry. 
And in the analysis of the compilation problem of finding circuit structures that are invariant under transformation according to the different types of symmetry present in a given topology. Particularly regarding topologies similar to those available on different types of existing quantum hardware. 
It would also be interesting to test how these methods combine with variational quantum algorithms. It is well-known that variational algorithms have some natural robustness to noise \cite{mcclean_theory_2016,omalley_scalable_2016,fontana_non-trivial_2022}, perhaps simplifying the noise either by symmetry or randomisation enhances this robustness. What the effect would be of imposing symmetries on the rotation gates within variational circuits in such a way that the error propagation maintains symmetry, and whether there is a trade-off between error simplification and restricting the expressivity of the circuit. 
Another area that could merit further investigation is how error complexity reductions combine with different error mitigation techniques \cite{yang_efficient_2022,berg_probabilistic_2022,mezher_mitigating_2022,temme_error_2017,endo_practical_2018,strikis_learning-based_2021}. Particularly in analysing how simplifying errors affects the scalability and effectiveness of the mitigation.

\text{ }
\section*{Acknowledgements}

We would like to thank Mina Doosti, Rawad Mezher, Robert Booth and Shelagh Casebourne for reading versions of this manuscript and providing insightful feedback. We would also like to thank Dominik Leichtle, Mina Doosti, Ross Grassie, Robert Booth and Theodoros Kapourniotis for interesting discussions; and Louise Smith for help with creating the diagrams. J.M. and D.S. are grateful for the support of grant funding from the School of Informatics at the University of Edinburgh.

\appendix
\vspace{2em}
\label{appendix}

\section{Randomisation parallelisation for a 2-qubit subset of a 3-qubit circuit \label{example_rand}}

A three qubit circuit is parallelised across $N$ different sets of three qubits, each of which is acted on by a different three qubit stochastic Pauli channel. A two qubit subset of the circuit is labelled by the tuple $q=(1,2)$, denoting the subset consisting of the first two qubits. The coefficients acting on this subset of qubits for the error channels of the different parallel circuits, $\{\mathcal{E}^\mathbf{G}_{i,q}\}_{i=1,\ldots,N}$, are denoted $c_{\mathbf{G},i}\in \{c_{\mathbf{G},1},\ldots, c_{\mathbf{G},N}\}$, where $\mathbf{G}= \{X,Y,Z\}^{\otimes 2}$, $c^{\mathbf{G}}_{i,q} \sim \mathcal{D}(\eta_q, [0,1])$, and $i \in {1,\ldots,N}$. Averaging over the coefficients of the parallel circuits gives the expected values of the coefficients relating to the effective channel. For example, averaging the $c_{X_1 X_2}$ operators across all $N$ parallel circuits gives an expected value of $$c_{X_1 X_2}'=\frac{1}{N}\sum_{i=1}^N c_{X_1 X_2}^i$$ where $i\in\{1,\ldots,N\}$ denotes the $i^{\text{th}}$ parallel circuit. The two qubit Pauli operators acting on this subset are $G_1 G_2 \in \{X_1 X_2, X_1 Y_2, X_1 Z_2, Y_1 X_2, Y_1 Y_2, Y_1 Z_2, Z_1 X_2, Z_1 Y_2, Z_1 Z_2\}$. The coefficient of operator ${G_1 G_2}$ averaged over the parallel circuits will converge to the expected value with increasing $N$, so that

\begin{equation}
\begin{split}
&\text{Pr}[(c_{G_1 G_2}' - \eta_{G_1 G_2})>\epsilon]\leq 2e^{-2N\epsilon^2}. \\
\end{split}
\end{equation}
Where $\eta_{G_1 G_2}$ is the mean of probabilities of the ${G_1 G_2}$ from all of the parallel circuits. Due to the shared distributions of each of the coefficients over the parallel circuits, combining them in the effective channel results in the averaged probabilities approaching the same value $\eta$, i.e. $\eta_{X_1 X_2}=\eta_{X_1 Y_2}=\ldots=\eta_{Z_1 Z_2}=\eta$.

\section{Quantum Hardware Details}

The experiments in Section \ref{app} were run on the Rigetti Aspen-M-3 superconducting quantum device.
This device consists of 79 superconducting transmon qubits arranged on an octagonal lattice. The two qubit native gates for these devices are $CZ$ and $XY$, and the single qubit native gates are $RZ(\theta)$ and $RX(\pi/2)$. By default, measurement is in the $Z$ basis and the initial $n$ qubit quantum state is $\ket{0}^{\otimes{n}}$.

\bibliography{Parallelisation.bib}

\end{document}